%% file: main.tex
\documentclass[prx,onecolumn,secnumarabic,nobibnotes,superscriptaddress,aps]{revtex4-2}
\input{preamble}

\usepackage{cellspace} 
\usepackage{makecell} 
\setcellgapes{3pt}

\usepackage{booktabs}

\newcommand{\DeptMath}{Department of Mathematics, University of California, Berkeley, CA 94720, USA}
\newcommand{\LBLMath}{Applied Mathematics and Computational Research Division, Lawrence Berkeley National Laboratory, Berkeley, CA 94720, USA}
\newcommand{\CIQC}{Challenge Institute of Quantum Computation, University
of California, Berkeley, CA 94720, USA}

\newcommand{\prepor}{$\mathrm{PREPARE}$\xspace}
\newcommand{\selector}{$\mathrm{SELECT}$\xspace}

\usepackage{soul}

\begin{document}

\title{Multi-level quantum signal processing with applications to ground state preparation using fast-forwarded Hamiltonian evolution}
\author{Yulong Dong}
\email{dongyl@berkeley.edu}
\affiliation{\DeptMath}
\author{Lin Lin}
\email{linlin@math.berkeley.edu}
\affiliation{\DeptMath}
\affiliation{\LBLMath}
\affiliation{\CIQC}
\date{\today}

\begin{abstract}
The preparation of the ground state of a Hamiltonian $H$ with a large spectral radius has applications in many areas such as electronic structure theory and quantum field theory. Given an initial state with a constant overlap with the ground state,
and assuming that the Hamiltonian $H$ can be efficiently simulated with an ideal fast-forwarding protocol,  we first demonstrate that employing a linear combination of unitaries (LCU) approach can prepare the ground state at a cost of $\mathcal{O}(\log^2(\|H\| \Delta^{-1}))$ queries to controlled Hamiltonian evolution. Here $\|H\|$ is  the spectral radius of $H$ and $\Delta$ the spectral gap. However, traditional Quantum Signal Processing (QSP)-based methods fail to capitalize on this efficient protocol, and its cost scales as $\mathcal{O}(\|H\| \Delta^{-1})$. To bridge this gap, we develop a multi-level QSP-based algorithm that exploits the fast-forwarding feature. This novel algorithm not only matches the efficiency of the LCU approach when an ideal fast-forwarding protocol is available, but also exceeds it with a reduced cost that scales as $\mathcal{O}(\log(\|H\| \Delta^{-1}))$. Additionally, our multi-level QSP method requires only $\mathcal{O}(\log(\|H\| \Delta^{-1}))$ coefficients for implementing single qubit rotations. This eliminates the need for constructing the PREPARE oracle in LCU, which prepares a state encoding $\mathcal{O}(\|H\| \Delta^{-1})$ coefficients regardless of whether the Hamiltonian can be fast-forwarded.

\end{abstract}

\maketitle

\section{Introduction}\label{sec:intro}
Recent years have witnessed considerable advancements in the development of quantum algorithms for scientific computing. This progress is fueled by novel construction of efficient quantum circuits, capable of representing a wide array of non-unitary matrices. They extend beyond the initial goal of simulating unitary dynamics, as envisioned by Benioff and Feynman \cite{Benioff1980, Feynman1982}, and introduce new possibilities for applications. Noteworthy developments include methods based on Quantum Phase Estimation, such as the HHL algorithm \cite{HarrowHassidimLloyd2009}, the Linear Combination of Unitaries (LCU)  \cite{BerryChildsCleveEtAl2014}, and methods based on the Quantum Signal Processing (QSP) \cite{LowChuang2017}, exemplified by the Quantum Singular Value Transformation (QSVT) \cite{GilyenSuLowEtAl2019, MartynRossiTanEtAl2021}. QSP-based algorithms are especially effective for representing various matrix functions with minimal ancillary qubit requirements.

The original formulation of QSP-based methods requires matrix access models using block encodings \cite{LowChuang2019}. It has been generalized to  utilize the Hamiltonian evolution input model, which can be implemented  through e.g., Trotter expansion methods~\cite{DongLinTong2022,SilvaBorgesAolita2022,MotlaghWiebe2023,BerryMotlaghPantaleoniEtAl2024}. Trotter based Hamiltonian simulation can be  easier to implement on near term hardware, and has some unique advantages such as commutator error scaling \cite{ChildsSuTranEtAl2021} and efficient simulation in the low-energy subspace \cite{SahinogluSomma2020,ZlokapaSomma2024}. Additionally, Hamiltonian simulation can be performed efficiently in the interaction picture \cite{LowWiebe2019}, and it allows for the possibility of fast-forwarded simulations \cite{BerryAhokasCleveEtAl2007,AtiaAharonov2017,GuSommaSahinoglu2021fast,AnLiuWangEtAl2022}.

When integrating advanced Hamiltonian simulation techniques with QSP-based methods, a significant challenge arises: QSP relies on polynomial transformations within the bounded interval $[-1,1]$. Many practical Hamiltonians in physics and chemistry \cite{AruteAryaBabbushEtAl2020,BauerBravyiMottaEtAl2020,LeeLeeZhaiEtAl2023} possess large spectral radii. Consequently, Hamiltonians must be rescaled to ensure their spectra fit within this interval. This rescaling process increases the polynomial degrees required to implement a matrix function, which reduces the efficiency of QSP-based methods. Conversely, LCU methods do not require such a rescaling operation. Denote the spectral radius of a Hamiltonian $H$ by $\norm{H}$. In the ideal fast-forwarding regime (see \cref{sec:problem} for the definition), the cost for implementing the Hamiltonian simulation can be independent of $\norm{H}$. Under such assumptions, the cost for
implementing the SELECT oracle in LCU, which directly queries the Hamiltonian evolution oracle, is only logarithmic on $\norm{H}$, while the cost for implementing QSP-based methods is linear in $\norm{H}$. This difference highlights a gap in our understanding of these quantum primitives, and a natural question follows:

\vspace{.5em}

\textit{Can we develop a quantum algorithm that can leverage fast-forwarded Hamiltonian evolution and bridge the performance gap between LCU-based and QSP-based methods?}

\vspace{.5em}

In this paper, we provide a positive answer to this question in one setting.
For the task of ground state preparation, we introduce a multi-level QSP-based method. This method requires only a logarithmic number of queries to the Hamiltonian evolution oracle in $\norm{H}$, while also significantly reducing the number of ancilla qubits.

To our knowledge, the complexity analysis of LCU-based methods for preparing the ground state using fast-forwarded Hamiltonian evolution has not appeared in the literature.  Although the SELECT oracle in LCU is very efficient, the PREPARE oracle, which is a unitary that prepares the Fourier coefficients, still incurs a cost that grows linearly with $\norm{H}$. Our new multi-level approach does not rely on such PREPARE oracles, and the gate complexity (for gates other than Hamiltonian simulation) provably scales logarithmically with $\norm{H}$ under idealized fast-forwarding conditions (see the main results in \cref{sec:main_results} and \cref{tab:compare_algs_state}).

The remainder of this paper is organized as follows: We begin with an overview of the problem statement and models for fast-forwarding in \cref{sec:problem}. Our primary findings are then highlighted in \cref{sec:main_results}. In \cref{sec:lcu}, we present the analysis of LCU-based method for ground state preparation using fast-forwarded Hamiltonian evolution. \cref{sec:ground-state} is dedicated to discussing QSP-based methods for ground state preparation utilizing fast-forwarded Hamiltonian evolution. The application of the standard QSP-based method is scrutinized in \cref{sec:standard-qsp}. We then lay the motivation of our algorithmic proposal through the lens of matrix polynomial in \cref{sec:connection-qsp}.  Progressing further, \cref{sec:multi-qsp} introduces our innovative multi-level QSP-based approach, aimed at overcoming the identified challenges. This section will also present a comprehensive complexity analysis, highlighting the method’s potential for improvements in time complexity and ancilla requirement. In \cref{sec:alpha-soft}, a detailed analysis when using another fast-forwarding model is presented. We conclude in \cref{sec:conclusion}, reflecting on the key achievements of our study and outlining promising directions for future research.

\section{Problem statement}\label{sec:problem}

Without loss of generality, let $H$ be a positive semidefinite Hermitian matrix with a large spectral radius $\norm{H}$. The Hamiltonian evolution with time $t$ is defined as $e^{- \I H t}$. The evolution of a general Hamiltonian is equivalent to that of a positive semidefinite one up to a global phase, or up to a time-dependent $Z$-rotation when querying the controlled Hamiltonian evolution. Denote the eigenpairs of the Hamiltonian $H$ as $\{ \lambda_k, \ket{\psi_k} \}$, ordered such that $\lambda_0 < \lambda_1 \le \lambda_2 \le \cdots$. Our objective is to prepare the ground state $\ket{\psi_0}$, starting from an initial guess $\ket{\phi}$ that overlaps with the ground state: $\abs{\braket{\phi | \psi_0}} = \gamma > 0$. The initial guess is prepared by an oracle $O_I$. Additionally, we assume the existence of a value $\mu$ and a predefined spectral gap $\Delta$, satisfying $\lambda_0 \le \mu - \Delta / 2 < \mu + \Delta / 2 \le \lambda_1$.  The parameter $\mu$ can be determined using e.g., a binary search strategy, and the cost of the algorithm is expected to depend inverse polynomially on $\gamma$ \cite{LinTong2020,DongLinTong2022}. 

Additionally, we assume that the quantum evolution of $H$ can be
\emph{fast-forwarded}~\cite{BerryAhokasCleveEtAl2007,AnLiuWangEtAl2022}. One
model for describing the fast-forwarded Hamiltonian evolution is to assume that
the Hamiltonian evolution can be implemented accurately within a cutoff time
$\tau$ with constant cost \cite{AtiaAharonov2017,GuSommaSahinoglu2021fast}. 
Another model is that the cost of the Hamiltonian evolution grows sublinearly with the evolution time.
For instance, in Ref. \cite{AnLiuWangEtAl2022}, the authors propose a model of
fast-forwarded Hamiltonian evolution where the implementation cost increases as
the square root of the evolution time. Both models can be described by the following definition of fast-forwarded Hamiltonian evolution.

\begin{defn}[Fast-forwarded Hamiltonian evolution]\label{def:tau-cutoff-fast-forwarding}
The quantum evolution of an $n$-qubit Hamiltonian $H$ can be fast-forwarded with parameter $(\tau, \alpha)$ if for any precision $\delta > 0$, there exists $0 \le \alpha \le 1$, $\tau > 0$ independent of $\norm{H}$, and a $t$-dependent unitary matrix $O_H(t)$ such that
\begin{enumerate}
\item $O_H(t)$ is accurate: 
\begin{equation}\label{eqn:HE}
    \sup_{\abs{t} \le \tau} \norm{O_H(t) - e^{- \I H t}} \le \delta,
\end{equation}
\item $O_H(t)$ is efficient: for any $\abs{t} \le \tau$ , the gate complexity
  for implementing $O_H(t)$ is upper bounded by $\Or(C(n) + (\abs{t} \cdot \norm{H})
  ^\alpha)$.
\end{enumerate}
\end{defn}
Here, the additive term $C(n)$ represents the cost for preparing the Hamiltonian into a fast forwardable form. For example, $C(n)=n^2$ when the Givens rotation is used to diagonalize a quadratic fermionic Hamiltonian~\cite{KivlichanMcCleanWiebeEtAl2018,GuSommaSahinoglu2021fast}. This term is independent of the spectral radius $\norm{H}$, which is the focus of this paper. Therefore when presenting complexity analysis, this additive term is omitted for simplicity. For the second term,  it is important to emphasize that the $\abs{t} \cdot \norm{H}$ dependency not
only serves to make the complexity dimensionless but also to prevent potential
shortcuts by rescaling the Hamiltonian. For instance, in scenarios where the
exponents differ, such as $\Or(\abs{t}^{\alpha_1} \norm{H}^{\alpha_2})$ with
$\alpha_2 > \alpha_1$, we can then simulate a scaled Hamiltonian $H/c$ to time $ct$. The resulting gate complexity becomes $\Or(c^{\alpha_1 - \alpha_2} \abs{t}^{\alpha_1} \norm{H}
^{\alpha_2})$. Taking the limit $c\to \infty$ leads to a vanishing cost, which
is a contradiction.

In this paper, we focus on two regimes of fast-forwarding. The first regime is
when the cutoff time $\tau$ is relatively small compared to the simulation time $t$. For simplicity, we assume the
short-time Hamiltonian evolution can be synthesized at a constant cost, which
means $\alpha = 0$. This case is referred to as \textit{$\tau$-cutoff
fast-forwarded Hamiltonian evolution}.  The second regime is when the cutoff time $\tau$
is large enough which is assumed to be larger than the intended simulation time
in the algorithm. Consequently, any Hamiltonian evolution needed in the
implementation of the algorithm requires only one query to $O_H$. This case is
referred to as \textit{$\alpha$-soft fast-forwarded Hamiltonian evolution}.
When $\tau$ is sufficiently large (longer than the simulation time $t$) and
$\alpha = 0$, \cref{def:tau-cutoff-fast-forwarding} means that the Hamiltonian
evolution with arbitrary evolution time can be implemented accurately at a
constant cost. This scenario is referred to as \textit{ideal fast-forwarding}.

\cref{def:tau-cutoff-fast-forwarding} suggests that algorithms based on fast-forwarded Hamiltonian evolution should minimize the number of queries by using long-time evolution whenever possible.
First, the error of $r$ queries to $O_H$ is upper bounded by $r\delta$ according to \cref{eqn:HE}. Hence, we need to minimize the number of queries to minimize the error growth in the algorithm. Second, to synthesize a time evolution of a given time $0 < t < \tau$, we may have a choice to implement it through two steps $O_H(c t) O_H((1 - c) t)$ for some $0 < c < 1$. However, due to the concavity of the gate complexity, we have $(ct)^\alpha + ((1 - c)t)^\alpha \ge t^\alpha$. Consequently, splitting the evolution time into smaller time steps increases the gate  complexity.

When the cutoff time $\tau$ is small relative to the simulation time $t$, long-time Hamiltonian evolution can be effectively synthesized as $e^{-\I t H} = O_H^r(t / r)$ by stitching together short-time evolutions satisfying $t/r\le \tau$. The cost of implementing this long-time evolution should grow linearly with the total evolution time $t$.

\begin{lem}[Long-time Hamiltonian evolution synthesis using fast-forwarded assumption]\label{lma:fast-forward}
    Based on \cref{eqn:HE}, implementing a Hamiltonian evolution with evolution time parameter $t$ requires $r := \lceil t / \tau \rceil$ queries to $O_H$ with error $r \delta$.
\end{lem}
\begin{proof}
    We will show that the Hamiltonian evolution to time $t$ can not be implemented when the number of queries to $O_H$ is strictly less than $r$. According to \cref{eqn:HE}, each query to $O_H$ can advance the simulation time by at most $\tau$. If there exists a protocol that implements $e^{- \I t H}$ with $\wt{r}$ queries to $O_H$, where $\wt{r}$ is strictly less than $r$, then the total propagation time is $\wt{r} \tau \le (r - 1) \tau$. For cases where $t / \tau \in \ZZ$, it holds that $r = t / \tau$, leading to $\wt{r} \tau < r \tau = t$. Conversely, for $t / \tau \notin \ZZ$, we find $r = 1 + \lfloor t / \tau \rfloor < 1 + t / \tau$, resulting in $\wt{r} \tau < t$. So the number of queries cannot be smaller than $r$ to reach the total evolution time $t$. The error in implementation is then estimated by employing both telescoping and the triangle inequality.
\end{proof}

To simplify the presentation of complexity analysis, we will not consider the implementation error of Hamiltonian evolution in the main text. In \cref{sec:HE-error}, we analyze the propagation of the error in the quantum circuit due to imperfect Hamiltonian evolution.

\section{Main results}\label{sec:main_results}
The main intuition of our algorithm comes from the observation that applying a
filter function effectively reduces the spectral radius of a Hamiltonian by
eliminating certain high-energy components. Given the assumption of
fast-forwarding, some filter functions are easy to implement. Our strategy
involves applying a sequence of easily implementable filter
functions across multiple levels. At each level of this process, the spectral
radius of the \emph{effective} Hamiltonian is  dyadically reduced. When the cutoff time is large, this process continues until the effective Hamiltonian contains only the ground-state component. However, when the cutoff time is small, the process continues until the effective spectral radius is bounded by a constant due to the limited efficiency of implementing long-time evolution for higher resolution. In the case of a small cutoff time, at the final step, we apply a clean-up filter function to the final effective Hamiltonian. The cost of implementing this final step only depends on the absolute spectral gap
$\Delta$ and is independent of the spectral radius of the original Hamiltonian
$\norm{H}$. This approach differs from the standard QSP-based method, implemented through Quantum Eigenvalue Transformation of Unitary Matrices with Real Polynomials (QETU) outlined in Ref.
\cite{DongLinTong2022}. In the standard QSP-based method, the filter function is directly applied to the original Hamiltonian, resulting in a polynomial of high degree. A comparison between the filter functions used in these methods is shown in \cref{fig:filters_compare}. It demonstrates that to achieve a similar filtering result, the polynomial degree required by our multi-level QSP-based method can be significantly smaller than that used in standard QSP-based and LCU-based methods. Our new method
circumvents the high computational costs due to large $\norm{H}$ by leveraging
fast-forwarded Hamiltonian evolution. A conceptual illustration of this process is depicted in \cref{fig:multi-level-visualize}. This illustration demonstrates the rapid elimination of high-energy components in
the initial state, resulting in an effectively small-radius Hamiltonian at the
final level. It is important to note that each step of the
filtering process does not need to be perfect in order to control the overall
error of the process. This impact on complexity is quantitatively analyzed in the following sections.

\begin{figure}[htbp]
    \centering
    \includegraphics[width = .5\textwidth]{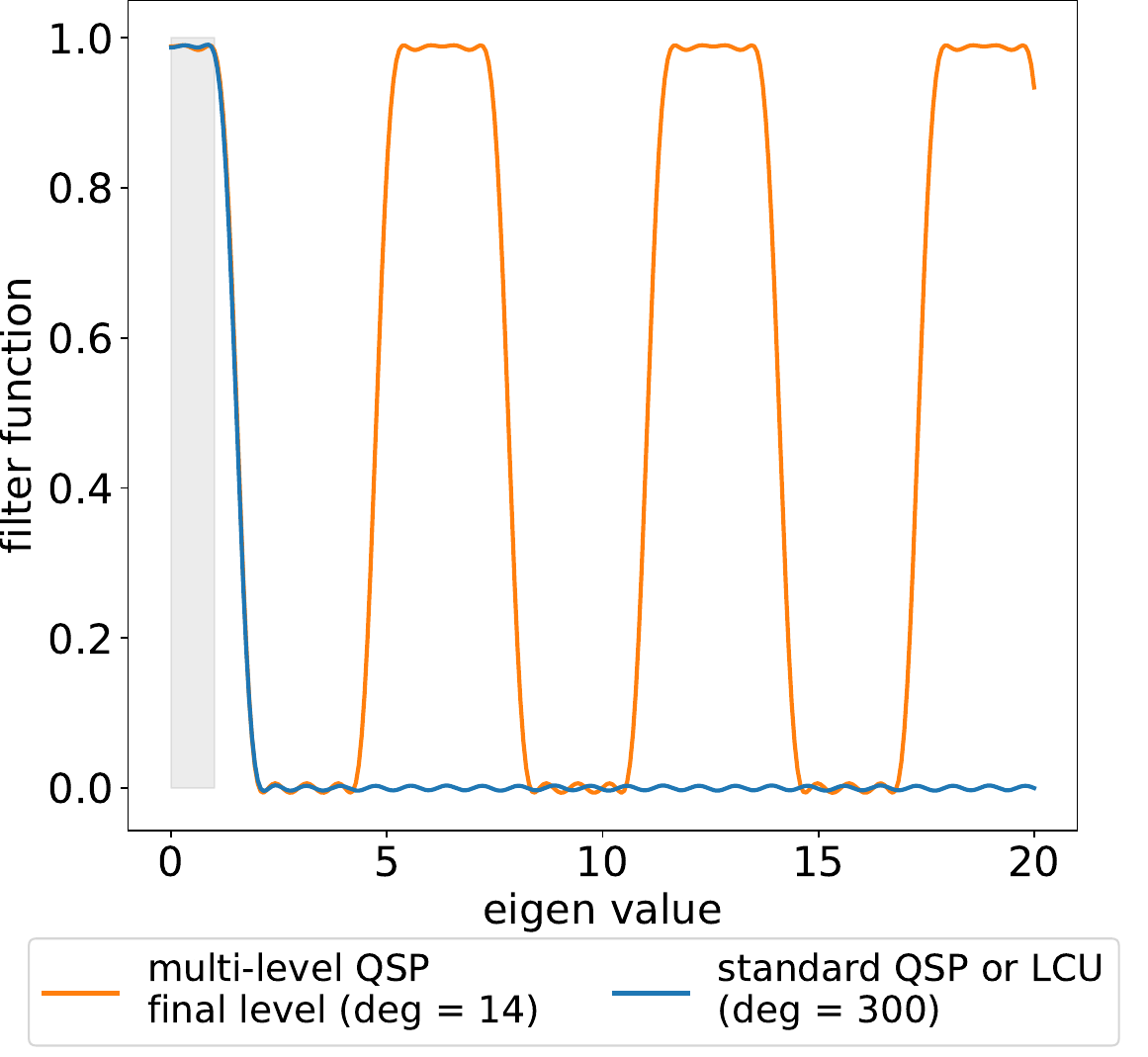}
    \caption{Comparing the filter functions used in the standard QSP-based method and at the final level of the multi-level QSP-based method. We assume the spectral radius of the Hamiltonian is $\norm{H} = 20$ and the spectral gap is $\Delta = 1$. The shaded area stands for the region containing only ground-state energy.}
    \label{fig:filters_compare}
\end{figure}

\begin{figure}[htbp]
    \centering
    \subfloat[]{\includegraphics[width = \textwidth]{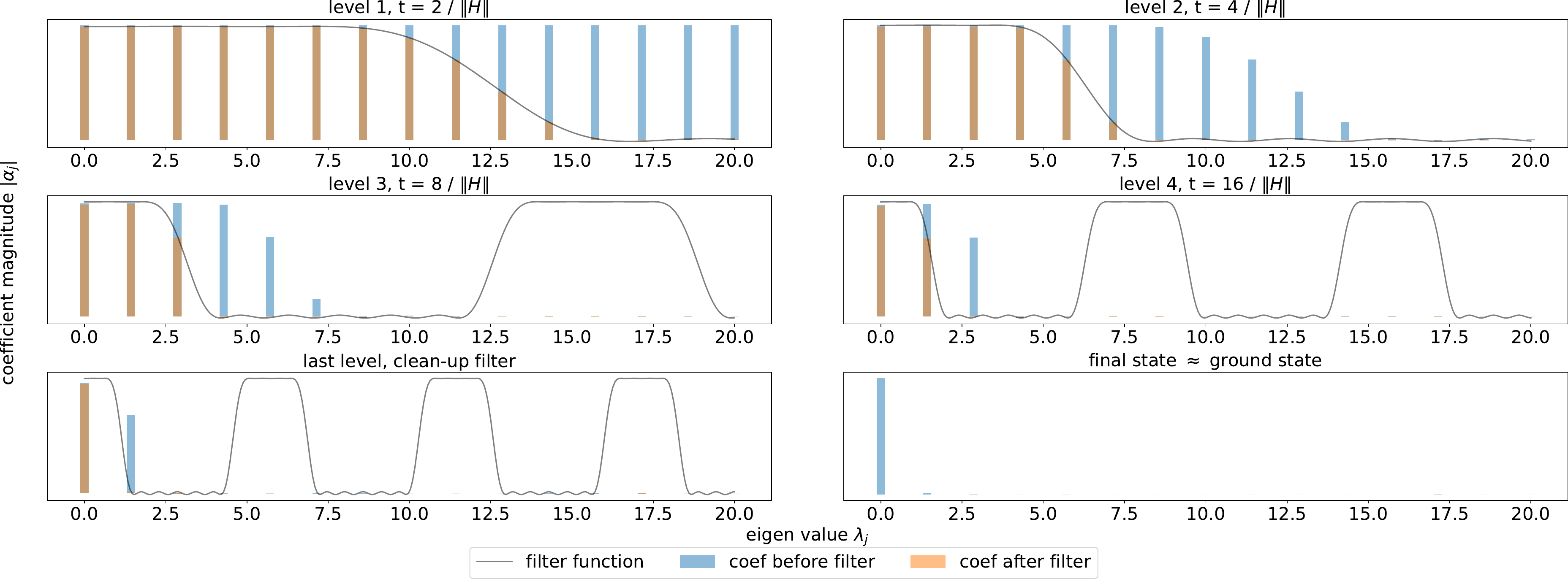}}

    \subfloat[]{\includegraphics[width = \textwidth]{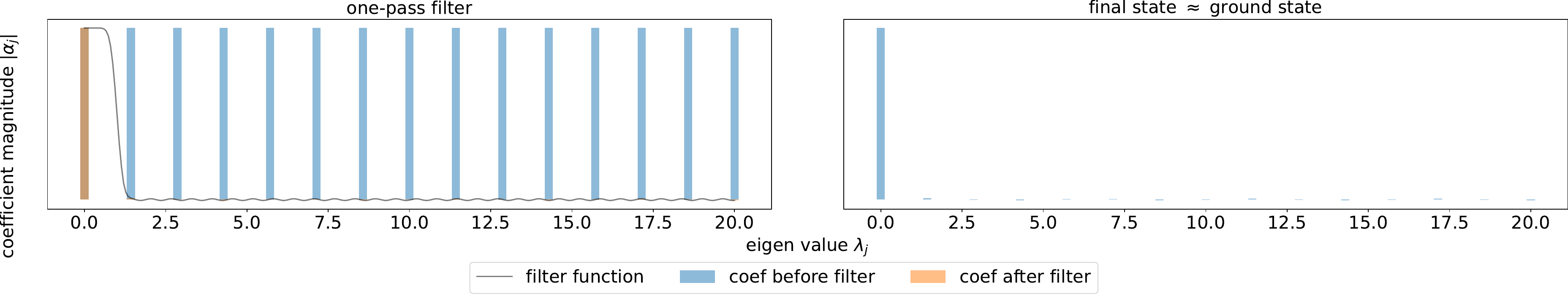}}
    \caption{An example illustrating methods for ground state preparation. Panel (a) visualizes the multi-level QSP-based method, while panel (b) displays the process using standard QSP-based or LCU-based methods. For simplicity, we consider a Hamiltonian with equally spaced eigen values whose spectral radius is $20$. The height of each bar stands for the magnitude of each coefficient with respect to the eigen basis before and after applying the filter function. }
    \label{fig:multi-level-visualize}
\end{figure}

In our implementation, each filter function is realized using QETU circuits, as outlined in Ref. \cite{DongLinTong2022}, with the specific circuit structure depicted in \cref{fig:qetu}. An implementation of ground state preparation utilizing intermediate measurements and post-selection is illustrated in \cref{fig:multi-level-meas}. The advantage of this circuit is that it does not involve multi-qubit control operations and the circuit depth is short. The corresponding limitation is the loss of coherence due to measurements. To prepare the ground state with a success probability $2/3$, the circuit needs to be run $\Or(\gamma^{-2})$ times. 

To improve the scaling with respect to $\gamma$, \cref{fig:multi-level} introduces a coherent quantum circuit implementation. We first present our complexity results using $\tau$-cutoff fast-forwarded Hamiltonian evolution. By incorporating a procedure called the compression gadget~\cite{LowWiebe2019,FangLinTong2023}, this circuit can be executed coherently with only $\Or(\log \log(\|H\|))$ additional ancilla qubits. By deferring measurements monitoring the success at each circuit level to the end of the circuit, the compression gadget allows a compact and coherent execution of the entire filtering process. Employing amplitude amplification, this coherent multi-level circuit achieves a quadratic speedup in terms of the initial overlap $\gamma$ from $\Or(\gamma^{-2})$ to $\Or(\gamma^{-1})$. This matches the lower bound for preparing the ground state under similar assumptions~\cite{LinTong2020,MandeWolf2023}.

\begin{figure}[htbp]
    \centering
    \subfloat[\label{fig:qetu}]{\includegraphics[width=\textwidth]{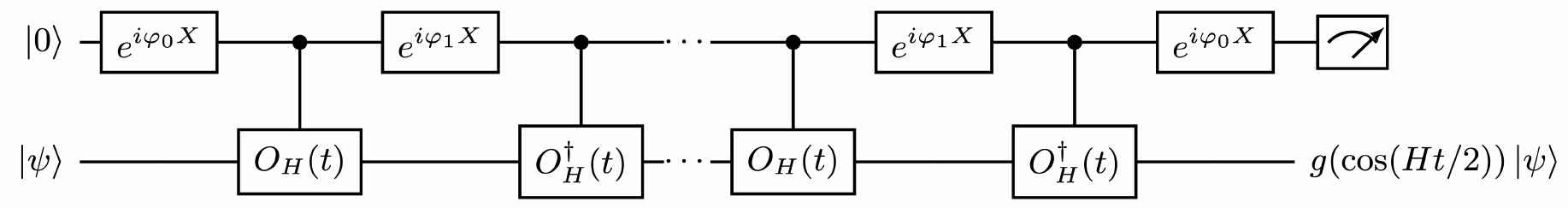}}

    \subfloat[\label{fig:multi-level-meas}]{\includegraphics[width=\textwidth]{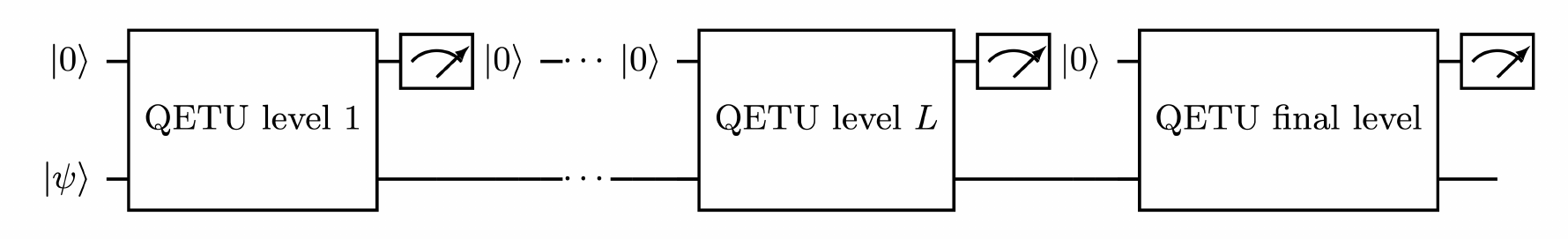}}

    \subfloat[\label{fig:multi-level}]{\includegraphics[width=\textwidth]{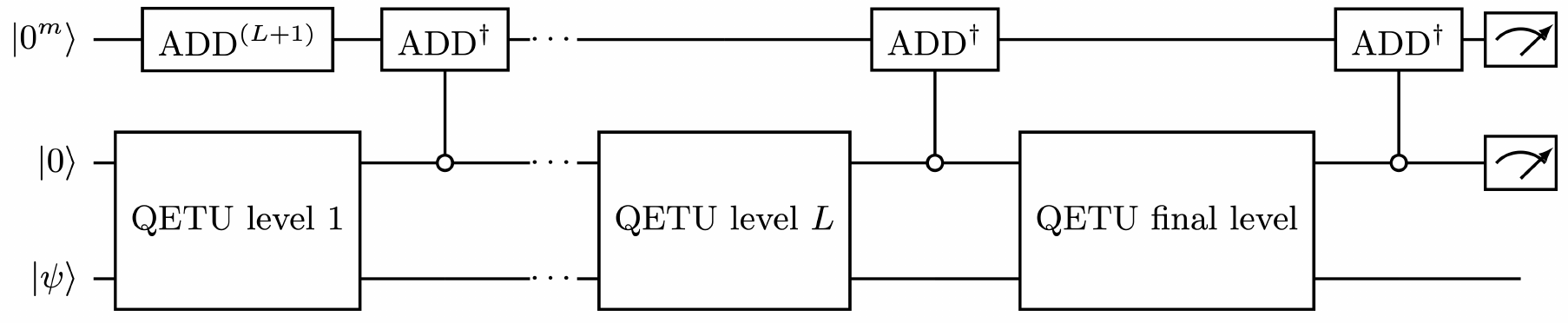}}
    
    \caption{Multi-level QSP-based quantum circuits. (a) QETU circuit which is used in implementing filter functions. (b) A multi-level circuit based on intermediate measurements and post-selection. (c) A coherent implementation of the multi-level circuit based on compression gadget. Here, $m = \lceil \log_2(L + 2) \rceil$ extra ancilla qubits are introduced to track the success of each level. }
    \label{fig:circuit}
\end{figure}

We present a detailed summary of the implementation cost for our algorithm and its comparison with other methods in \cref{tab:compare_algs_state}. Our proposed algorithms demonstrate superior performance over all other methods examined. While a key subroutine of our approach is QETU, we significantly reduce the cost with respect to the spectral radius by adopting a multi-level strategy. In comparison to the LCU approach, in the case of ideal fast-forwarding where $\alpha = 0$ and $\tau$ is large, our method achieves logarithmic gate complexity in terms of the spectral radius and the spectral gap, offering a quadratic improvement over the gate complexity of the SELECT oracle in LCU. Additionally, our method shows an exponential reduction in terms of $\norm{H}$ and $\Delta$ in the requirement for ancilla qubits compared to LCU, and our requirement on the number of ancilla qubits is also independent with the precision parameter $\epsilon$. It is also important to note that LCU requires a \prepor oracle which encodes target Fourier coefficients, the implementation of which can be notably resource-intensive. For example, in the case involving $d$ coefficients lacking a specific structure, a leading-edge solution outlined in Ref. \cite{BabbushGidneyBerryEtAl2018} requires $\Or(d + \log(\epsilon^{-1}))$ additional T gates and $\Or(\log(d \epsilon^{-1}))$ extra ancilla qubits.  When dealing with a large spectral radius, which implies a sharp filter function and a large $d$, the cost of this oracle could substantially increase and impact the overall complexity adversely. This is reflected in the complexity results of LCU-based method in \cref{tab:compare_algs_state}. Due to the costly \prepor oracle, which does not benefit from fast-forwarding, the overall gate complexity of the LCU-based method scales linearly with $\norm{H}$ for any $0 \le \alpha \le 1$. In contrast, our multi-level QSP-based method achieves a speedup, manifesting as either a logarithmic dependency on $\log(\norm{H})$ when $\alpha = 0$ or a sublinear scaling of $\norm{H}^\alpha$ when $0 < \alpha < 1$. This enhancement is attributed to the multi-level filtering strategy that combines the efficiency of the \selector oracle of LCU and the flexibility of QSP.

\begin{table}[th]
\label{tab:ground_state_prep}
    \centering
    \makegapedcells
        \scalebox{.8}{\begin{tabular}{p{2.5cm}|c|c|c}
        \hline
\hline
                                                                                             & This work, multi-level QSP                                                                                                                                                                              & LCU                                                                                                                                                                                                                                                                                                                            & Standard QSP                                                                                                                                     \\ \hline
\begin{tabular}[c]{@{}c@{}}gate complexity\\ case 1: \\ (small $\tau$, $\alpha = 0$)\end{tabular}       & \begin{tabular}[c]{@{}c@{}}$\wt{\Or}\left(\gamma^{-1} \left(\log(\norm{H}) + (1+\tau^{-1})\Delta^{-1}  \right) \log\left(\epsilon^{-1}\right) \right)$ 
\end{tabular} & \begin{tabular}[c]{@{}c@{}}$\wt{\Or}\left(\gamma^{-1} \log^2(\norm{H}\Delta^{-1} \log(\epsilon^{-1}))\right.$\\ $\left.+  \tau^{-1} \gamma^{-1} \Delta^{-1} \log\left( \norm{H} \right) \log\left(\epsilon^{-1}\right) \right)$\\ + $\wt{\Or}(\gamma^{-1} \norm{H} \Delta^{-1} \log(\epsilon^{-1}))$\end{tabular} & \begin{tabular}[c]{@{}c@{}}$\wt{\Or}(\max\{1, \tau^{-1} \norm{H}^{-1}\} $\\ $\gamma^{-1} \Delta^{-1} \norm{H} \log(\epsilon^{-1}))$\end{tabular} \\ \hline
\begin{tabular}[c]{@{}c@{}}gate complexity\\ case 2: \\ (large $\tau$, $\alpha = 0$)\end{tabular}       & $\wt{\Or}(\gamma^{-1}\log(\norm{H} \Delta^{-1}) \log(\epsilon^{-1}))$                                                                                                                               & \begin{tabular}[c]{@{}c@{}}$\wt{\Or}(\gamma^{-1}\log^2(\norm{H} \Delta^{-1} \log(\epsilon^{-1})))$\\ + $\wt{\Or}(\gamma^{-1} \norm{H} \Delta^{-1} \log(\epsilon^{-1}))$\end{tabular}                                                                                                                                           & $\wt{\Or}(\gamma^{-1} \norm{H} \Delta^{-1} \log(\epsilon^{-1}))$                                                                                 \\ \hline
\begin{tabular}[c]{@{}c@{}}gate complexity\\ case 3: \\ (large $\tau$, $0 < \alpha \le 1$)\end{tabular} & $\wt{\Or}(\gamma^{-1} \norm{H}^\alpha \Delta^{-\alpha} \log(\epsilon^{-1}))$                                                                                                                         & \begin{tabular}[c]{@{}c@{}}$\wt{\Or}(\gamma^{-1} \norm{H}^\alpha \Delta^{-\alpha} \log^\alpha(\epsilon^{-1}))$\\ + $\wt{\Or}(\gamma^{-1} \norm{H} \Delta^{-1} \log(\epsilon^{-1}))$\end{tabular}                                                                                                                               & $\wt{\Or}(\gamma^{-1} \norm{H} \Delta^{-1} \log(\epsilon^{-1}))$                                                                                 \\ \hline
\# ancilla qubits                                                                            & \begin{tabular}[c]{@{}c@{}} $\Or(\log \log (\norm{H}))$ (case 1)\\ $\Or(\log \log (\norm{H} \Delta^{-1}))$ (case 2, 3) \end{tabular}                                                                                                                                                                             & $\wt{\Or}(\log(\norm{H} \Delta^{-1} \gamma^{-1} \epsilon^{-1}))$                                                                                                                                                                                                                                                                            & $\Or(1)$                                                                                                                                         \\ \hline \hline
        \end{tabular}}
    \caption{Comparison of the performance of quantum algorithms for ground state preparation in terms of gate complexity, and the number of ancilla qubits. The Hamiltonian evolution is assumed to be fast-forwarded with parameter $(\tau, \alpha)$ (see \cref{def:tau-cutoff-fast-forwarding} for details). $\gamma$ is the overlap between the initial guess $\ket{\psi_0}$ and the ground state, $\Delta$ is a lower bound of the spectral gap, and $\norm{H}$ is the spectral radius of the Hamiltonian. The gate complexity of LCU-based method is expressed in two parts. The first term in big-O notation stands for the implementation cost of the \selector oracle, and the second term corresponds to the cost of the \prepor oracle.}
    \label{tab:compare_algs_state}
\end{table}

\section{Ground state preparation methods using  fast-forwarded Hamiltonian evolution with cutoff and linear combination of unitaries}\label{sec:lcu}

The Linear Combination of Unitaries (LCU) is a quantum algorithmic technique that implements a complex quantum operator by decomposing it into simpler, easy-to-implement unitary operations \cite{ChildsKothariSomma2017}. Typical choices of these unitary operations include Hamiltonian evolutions, executed through efficient Hamiltonian simulation algorithms, and Chebyshev polynomials, realized via quantum walks. The process employs a \prepor oracle to linearly combine these unitary elements. Consequently, designing a quantum algorithm to solve a given problem essentially becomes a task of finding the best approximation, whether through Fourier or Chebyshev series. When integrating $M$ terms in LCU, the \prepor oracle uses multi-qubit control gates and  $\Or(\log_2(M))$ ancilla qubits. Under the assumption of fast-forwarded quantum evolution, the former option, utilizing Hamiltonian evolution, is more suited for ground state preparation tasks. Essentially, this involves the implementation of a filter function designed to eliminate all excited states. The filter function that implements the ground-state projection can be expressed in the context of LCU as follows:
\begin{equation}
    f(x) = \sum_{k = 1}^M c_k e^{- \I x t_k},\ f(H) = \sum_{k = 1}^M c_k e^{- \I H t_k}.
\end{equation}
The function $f(x)$ satisfies the following conditions:
\begin{enumerate}
    \item[(A1)] $f(x) \approx 1$ for any $x < \mu - \Delta / 2$,
    \item[(A2)] $f(x) \approx 0$ for any $\mu + \Delta / 2 < x < \norm{H}$,
    \item[(A3)] and $\abs{f(x)} \le 1$ for any $0 \le x \le \norm{H}$.
\end{enumerate}

Approximating Heaviside function has been a well studied problem in the context of quantum algorithms. In \cite{LowChuang2017}, an approximation utilizing the error function is suggested. Following this, \cite{LinTong2022} offers an approximation by integrating a smeared Dirac function, contributing to the methodology of constructing such approximations. Further refining the complexity of the approximation found in \cite{LinTong2022}, Ref. \cite{WanBertaCampbell2022} presents a Fourier series that approximates the Heaviside function with explicitly defined coefficients. Our analysis begins with an approximation of the filter function needed for the LCU-based method, as derived from \cite[Lemma 1]{WanBertaCampbell2022}.

\begin{thm}[Fourier approximation to the LCU-based filter function]\label{thm:Fourier-approx}
    When $\Delta / \norm{H} < \pi / 2$, for any $\epsilon > 0$,  there is a Fourier series $f(x) = \sum_{k \in \mc{I}} c_k e^{- \I t_k x}$ where $\mc{I} = \{0\} \cup \{\pm (2k + 1)\}_{k = 0}^d$, $t_k = k / \norm{H}$, and $d = \Or(\norm{H} \Delta^{-1} \log(\epsilon^{-1}))$, which satisfies conditions (A1)-(A3) and $\norm{c}_1 := \sum_{k \in \mc{I}} \abs{c_k} = \Or(\log(d))$.
\end{thm}
\begin{proof}
    According to \cite[Lemma 1]{WanBertaCampbell2022}, for any $0 < \delta < \pi / 2$ and $\epsilon > 0$, a Fourier approximation to the Heaviside function can be derived as $\wt{f}(x) = \sum_{k \in \mc{I}} \wt{c}_k e^{- \I k x}$ where $d = \Or(\delta^{-1}\log(\epsilon^{-1}))$ which satisfies $\max_{\delta < \abs{x} < \pi - \delta} \abs{\wt{f}(x) - H(x)} \le \epsilon$, $\norm{\wt{f}}_\infty \le 1$, and $\norm{\wt{c}}_1 = \Or(\log(d))$. Let us choose $\delta = \Delta / (2 \norm{H}) \le 1 / 2$. Then, the Fourier series $f(x) := \wt{f}((x - \mu) / \norm{H})$ satisfies these conditions. The Fourier expansion is written as $f(x) = \sum_{k \in \mc{I}} \wt{c}_k e^{\I t_k \mu} e^{- \I t_k x}$. Hence, the $1$-norm of the Fourier coefficients is also bounded as $\norm{c}_1 = \norm{\wt{c}}_1 = \Or(\log(d))$. 
\end{proof}

Utilizing the proposed filter function, the matrix function $f(H) = \sum_{k \in \mc{I}} c_k e^{-\I H t_k}$ represents a series of Hamiltonian evolution matrices with varying time parameters. This sum can be efficiently implemented through the LCU technique, which employs a \selector oracle to orchestrate a set of controlled time evolution matrices $\{e^{\pm \I 2^\ell H / \norm{H}} : \ell \le \lceil \log_2(2d+1) \rceil\}$. Additionally, a \prepor oracle is utilized to merge these matrices into a linear combination with coefficients $c$. This method requires a minimum of $\Or(\log(d))$ extra ancilla qubits to encode $\norm{c}_1^{-1} f(H)$, where the prefactor $\norm{c}_1^{-1}$ is referred to as a $1$-norm scaling constant. The efficiency of the \selector oracle’s design is significantly enhanced by the fast-forwarding assumption. According to \cref{lma:fast-forward}, the implementation cost for the \selector oracle, under the fast-forwarding assumption, is
\begin{equation}
    2 \sum_{\ell = 0}^{\lceil \log_2(2d+1)\rceil} \left\lceil 2^\ell / (\tau \norm{H}) \right\rceil \le 2 \sum_{\ell = 0}^{\lceil \log_2(2d+1)\rceil} \left(1 +  2^\ell / (\tau \norm{H})\right) = \Or\left(\log(d) + d / (\tau \norm{H})\right).
\end{equation}

The $1$-norm scaling constant inherent in the LCU method, combined with the initial overlap, sets the success probability of a single LCU run at $\gamma^2 / \norm{c}_1^2$. Consequently, to achieve a constant success probability, we employ amplitude amplification, which adjusts the overall cost of the LCU-based method. Thus, the total expense of achieving a desired success rate through amplitude amplification in the LCU framework is 
\begin{equation}
\begin{split}
    &\Or(\norm{c}_1 \gamma^{-1}) \times \Or\left(\log(d) + d / (\tau \norm{H})\right)\\
    &= \wt{\Or}\left(\gamma^{-1} \log^2(\norm{H}\Delta^{-1} \log(\epsilon^{-1})) +  \tau^{-1} \gamma^{-1} \Delta^{-1} \log\left( \norm{H} \right) \log\left(\epsilon^{-1}\right) \right)
\end{split}
\end{equation}
which exhibits only a weak dependence on the spectral radius $\norm{H}$ through a logarithmic term. Here, the polynomial degree is $d = \Or(\norm{H} \Delta^{-1} \log(\tilde{\epsilon}^{-1}))$, where the approximation error of the matrix function is set to $\tilde{\epsilon} = \norm{c}_1^{-1} \gamma \epsilon$, which is smaller than the target precision $\epsilon$ to account for the normalization.

However, one notable limitation of employing the LCU method is the complexity involved in implementing the \prepor oracle, which is responsible for managing the linear combination coefficients $\{\alpha_k\}$. The necessity for additional ancilla qubits and complex multi-qubit control gates particularly complicates its deployment on early fault-tolerant quantum devices. In the following analysis, we delve into the analysis of the implementation costs associated with the \prepor oracle.

The \prepor oracle can be constructed by using the method in Ref. \cite{BabbushGidneyBerryEtAl2018}. It leverages quantum read-only memory (QROM) to encode a finite-bit approximation of each coefficient onto an ancilla register, facilitating data retrieval for the oracle's operation. This approach improves the complexity efficiency compared to other methods \cite{SandersLowSchererEtAl2019} by observing that any temporary irrelevant data held in the ancilla register commutes with the \selector oracle, which does not interfere with the computation. By emulating an alias sampling method within the temporary register, the desired quantum state is prepared directly, eliminating the need for amplitude amplification. The implementation of a \prepor oracle incurs an overall T gate complexity of $\Or(d + \log(\tilde{\epsilon}^{-1}))$. This requires an additional $\Or(\log(d) + \log(\tilde{\epsilon}^{-1}))$ ancilla qubits. The rationale for using T gate count as a measure of complexity stems from the substantial time and space resources required for implementing error corrected T gates. Consequently, the T gate count can be a faithful measure of the overall runtime of an algorithm in fault-tolerant quantum computing \cite{BabbushGidneyBerryEtAl2018}. Given that the error from implementing the \prepor oracle can be subsumed into the overall approximation error of $f$, we simplify our discussion by not separately accounting for these errors. Consequently, incorporating the \prepor oracle into the LCU framework introduces an additional T gate complexity of $\wt{\Or}(\norm{H} \Delta^{-1} \log(\gamma^{-1} \epsilon^{-1}))$, where the $\wt{\Or}$ notation conceals further logarithmic factors attributable to adjusted approximation error $\tilde{\epsilon}$. This also requires extra $\Or(\log(d) + \log(\tilde{\epsilon}^{-1}))$ ancilla qubits. Using \cref{thm:Fourier-approx}, we deduce that implementing \prepor oracle requires $\Or(\log(\norm{H} \Delta^{-1} \log(\tilde{\epsilon}^{-1})) + \log(\tilde{\epsilon}^{-1})) = \wt{\Or}(\log(\norm{H} \Delta^{-1}\gamma^{-1} \epsilon^{-1}))$ extra ancilla qubits, where the $\wt{\Or}$ notation conceals a further logarithmic factors in the logarithm.  This method accommodates the implementation despite the absence of a structured pattern in the coefficients. However, it is worth noting that certain strategies, such as the one discussed in \cite{McArdleGilyenBerta2022}, can not be directly applicable due to these structural limitations in the coefficients.

We summarize the above analysis as the following theorem.
\begin{thm}[LCU-based ground state preparation using fast-forwarded Hamiltonian evolution]\label{thm:lcu-prepare}
    Using an LCU-based method, we can prepare the ground state, with probability $2/3$, up to fidelity $1 - \epsilon$, with the following cost:
    \begin{enumerate}
        \item[(1)] $\Or(\gamma^{-1})$ queries to $O_I$,
        \item[(2)] $\wt{\Or}\left(\gamma^{-1} \log^2(\norm{H}\Delta^{-1} \log(\epsilon^{-1})) +  \tau^{-1} \gamma^{-1} \Delta^{-1} \log\left( \norm{H} \right) \log\left(\epsilon^{-1}\right) \right)$ gates to implement \selector oracle,
        \item[(3)] $\wt{\Or}(\gamma^{-1} \norm{H} \Delta^{-1} \log(\epsilon^{-1}))$ additional T gates to implement \prepor oracle,
        \item[(4)] $\wt{\Or}(\log(\norm{H} \Delta^{-1} \gamma^{-1} \epsilon^{-1}))$ ancilla qubits.
    \end{enumerate}
\end{thm}

In the case of ideal fast-forwarding where $\tau$ is large enough, the implementation cost of \selector oracle is $\wt{\Or}(\mathrm{poly}\log(\norm{H}\Delta^{-1}))$ which depends weakly in both spectral radius and spectral gap in terms of squared logarithm. However, when $\tau$ is small, the implementation cost of \selector oracle is $\wt{\Or}(\Delta^{-1} \log(\norm{H}))$ in which the dependency of spectral gap is no longer logarithmic despite the weak logarithmic dependency of spectral radius.

\section{Ground state preparation methods using fast-forwarded Hamiltonian evolution with cutoff and quantum signal processing}\label{sec:ground-state}

\subsection{Standard QSP-based method }\label{sec:standard-qsp}

The approach to ground state preparation utilizing standard QSP methods, such as the \QETU method in Ref. \cite{DongLinTong2022}, applies a polynomial of $\cos(H/(2\norm{H}))$ to a \emph{normalized} Hamiltonian $H / \norm{H}$. The polynomial degree is inversely proportional to the \emph{relative} spectral gap $\Delta / \norm{H}$. This method queries the (controlled) Hamiltonian evolution $e^{-\I H / \norm{H}}$ with an evolution time of $1 / \norm{H}$. According to \cref{lma:fast-forward}, the number of queries to $O_H$ for constructing this building block  is $\max\{1, \lceil (\tau \norm{H})^{-1}\rceil\}$, introducing a multiplicative factor to the algorithm's overall complexity. Specifically, even when $H$ can be fast-forwarded to a time $\tau$ significantly greater than $\norm{H}^{-1}$, QSP is unable to leverage this capability and must continue to employ a small time step of $\norm{H}^{-1}$ for polynomial construction.
The cost of the ground state preparation based on standard QSP is determined according to \cite[Theorem 11]{DongLinTong2022}, which is summarized below.

\begin{thm}[QSP-based ground state preparation]\label{thm:standard-qsp}
    Using a QSP-based method, we can prepare the ground state, with probability $2/3$, up to fidelity $1 - \epsilon$, with the following cost:
    \begin{enumerate}
        \item[(1)] $\Or(\gamma^{-1})$ queries to $O_I$,
        \item[(2)] $\wt{\Or}(\max\{1, \tau^{-1} \norm{H}^{-1}\} \gamma^{-1} \Delta^{-1} \norm{H} \log(\epsilon^{-1}))$ gates,
        \item[(3)] two ancilla qubits.
    \end{enumerate}
\end{thm}

Two additional ancilla qubits are utilized: one for QETU implementation and another for amplitude amplification. To assess the impact of the multiplicative factor due to fast-forwarded Hamiltonian evolution, without loss of generality we assume $\norm{H}^{-1}\le\tau$, which means that the cost of implementing $e^{-\I H / \norm{H}}$ is equal to that of a single query to $O_H$. When $\norm{H}$ is relatively small, the QETU method \cite{DongLinTong2022} is an efficient method for preparing the ground state. However, when $\norm{H}$ is large,  the relative spectral gap $\Delta / \norm{H}$ becomes very small, and the overall complexity that increases linearly with $\norm{H}$. This is significantly more costly than the \selector oracle of the LCU approach which is able to take advantage of the fast forwarded Hamiltonian simulation.

\subsection{Motivation for the multi-level QSP based method}\label{sec:connection-qsp}

QSP can be understood as a nonlinear polynomial approximation framework. It effectively maps a set of parameters, termed phase factors, to a polynomial function under minimal constraints. Central to QSP-based methods is the concept of qubitization -- a two-dimensional subspace that bridges the approximation theory in $\mathrm{SU}(2)$ with high-dimensional matrix function transformations. This connection is particularly evident in QETU. Within the eigenbasis of the Hamiltonian, a QETU circuit performs interleaved $X$- and $Z$-rotations across a series of two-dimensional subspaces, each determined by the eigenvalues of the Hamiltonian. As  in \cite[Theorem 1]{DongLinTong2022}, this interleaved sequence of rotations facilitates a degree-$d$ polynomial transformation through $g$
\begin{equation}
    e^{\I \theta_0 X} e^{\I \lambda t Z} e^{\I \theta_1 X} e^{\I \lambda t Z} \cdots e^{\I \lambda t Z} e^{\I \theta_1 X} e^{\I \lambda t Z} e^{\I \theta_0 X} \ket{0} = g(\cos(\lambda t / 2)) \ket{0} + (*) \ket{1}
\end{equation}
where the parameter $t$ accounts for the evolution time of the Hamiltonian evolution. Here, the second term on the right-hand side represents an irrelevant unnormalized state. We refer to the detailed derivation in \cref{app:qetu-derivation}.

From the perspective of polynomial transformation, multi-level QETU facilitates the implementation of a product of polynomials with varying evolution times. As outlined in \cref{sec:multi-qsp}, on the eigen subspace corresponding to an eigenvalue $\lambda$, its action can be described as:
\begin{equation*}
    G(\cos(\lambda / \norm{H})) \ket{0} + (*) \ket{1} \text{ where } G(\cos( \lambda / \norm{H})) := \prod_{\ell = 1}^L g(\cos(2^{\ell-1} \lambda / \norm{H}))
\end{equation*}
The function $G$ is expressed as
\begin{equation*}
    G(x) = g(x) g(T_2(x)) g(T_4(x)) \cdots g(T_{2^{L-1}}(x))
\end{equation*}
where $T_n(x)$ stands for Chebyshev polynomials of the first kind. Consequently, $G$ is a real even polynomial of degree $(2^L - 1) d$. A direct implementation using the original QETU circuit requires $\Or(2^L d)$ queries to $e^{- \I 2 H / \norm{H}}$. Notably, the query complexity exhibits an exponential dependence on $L$ and, as a result, a linear dependence on the spectral radius $\norm{H}$. The substantial improvement in the dependence on $\norm{H}$ is attributed to a strategic implementation of $g\circ T_{2^\ell}$, taking advantage of the accelerated dynamics offered by fast-forwarded Hamiltonian evolution. As derived in \cref{sec:multi-qsp}, the multi-level strategy achieves an exponential reduction in the dependence in $\norm{H}$ which also outperforms LCU-based method in terms of a additive overall complexity (see \cref{thm:multi-level-aa}).

In \cref{fig:multi-level-visualize}, we showcase the construction of the filter function employing a multi-level strategy. It is evident that the ground-state component is consistently maintained across multiple levels, whereas the other regions are almost entirely suppressed through the application of multiple filters. This contrasts with the approach in standard QETU, which requires a very sharp filter and imposes stringent resource requirements.

\subsection{Multi-level QSP-based method}\label{sec:multi-qsp}

As detailed in the preceding subsection, employing QSP-based methods directly incurs a cost linearly dependent on the spectral radius $\norm{H}$. This limitation is mainly because the method can only query $e^{- \I H t}$ for a short time $t=1/\norm{H}$, which does not fully exploit the potential of fast-forwarded Hamiltonian evolution. In comparison, the key to the logarithmic dependence observed in LCU-based methods lies in the innovative construction of the \selector oracle, which harnesses the capabilities of fast-forwarded Hamiltonian evolution. To address this challenge, we propose a multi-level \QETU strategy in this subsection. This approach markedly improves efficiency by exponentially reducing the cost with respect to the spectral radius.

Let us consider a real even polynomial $g$ of degree $d$ which satisfies the following conditions:

\begin{enumerate}
    \item[(B1)] $\abs{g(x) - 1} \le \epsilon^\prime$ for any $x > \cos(\pi / 8) \approx 0.92$,
    \item[(B2)] $\abs{g(x)} \le \epsilon^\prime$ for any $0 < x < \cos(\pi / 4) \approx 0.71$,
    \item[(B3)] and $\abs{g(x)} \le 1$ for any $x \in [-1, 1]$.
\end{enumerate}

This polynomial can be constructed following the recipe in Ref. \cite{WanBertaCampbell2022} with complexity $d = \Or(\log(1/\epsilon^\prime))$. According to \cite[Theorem 1]{DongLinTong2022}, by performing a controlled implementation of $e^{- \I H t}$ for $d$ times, the \QETU method is capable of implementing $g(\cos(H t / 2))$. The difference between functions $f$ defined in conditions (A1-3) and $g$ defined in conditions (B1-3) arises from the intermediate cosine transformation, which reverses their monotonicity. It is important to highlight that the criteria set forth by conditions (B1-3) are markedly less stringent than those specified by conditions (A1-3). Notably, the transition gap from zero to one under conditions (B1-3) remains constant, which is independent of $\Delta$ or $\norm{H}$ as in (A1-3). As a result, the complexity involved in approximating the filter function under conditions (B1-3) demonstrates merely a logarithmic dependence on the precision parameter.  

In \cref{fig:approx}, we numerically illustrate the desired filtering function using various polynomial degrees. These polynomials are obtained through a convex-optimization-based method \cite{DongLinTong2022}, representing the numerically optimal polynomial approximations that meet the specified conditions. Notably, we observe that the function adheres closely to the required conditions with minimal error, even at lower degrees. Although the shaded area between intervals with restrictions remains uncontrolled due to the absence of specific constraints, the boundedness condition mitigates concerns about this region in the sequential multi-level applications. Our analysis will further demonstrate that the fuzziness in this middle region does not impact the overall performance of the algorithm.

\begin{figure}[htbp]
    \centering
    \includegraphics[width=.5\textwidth]{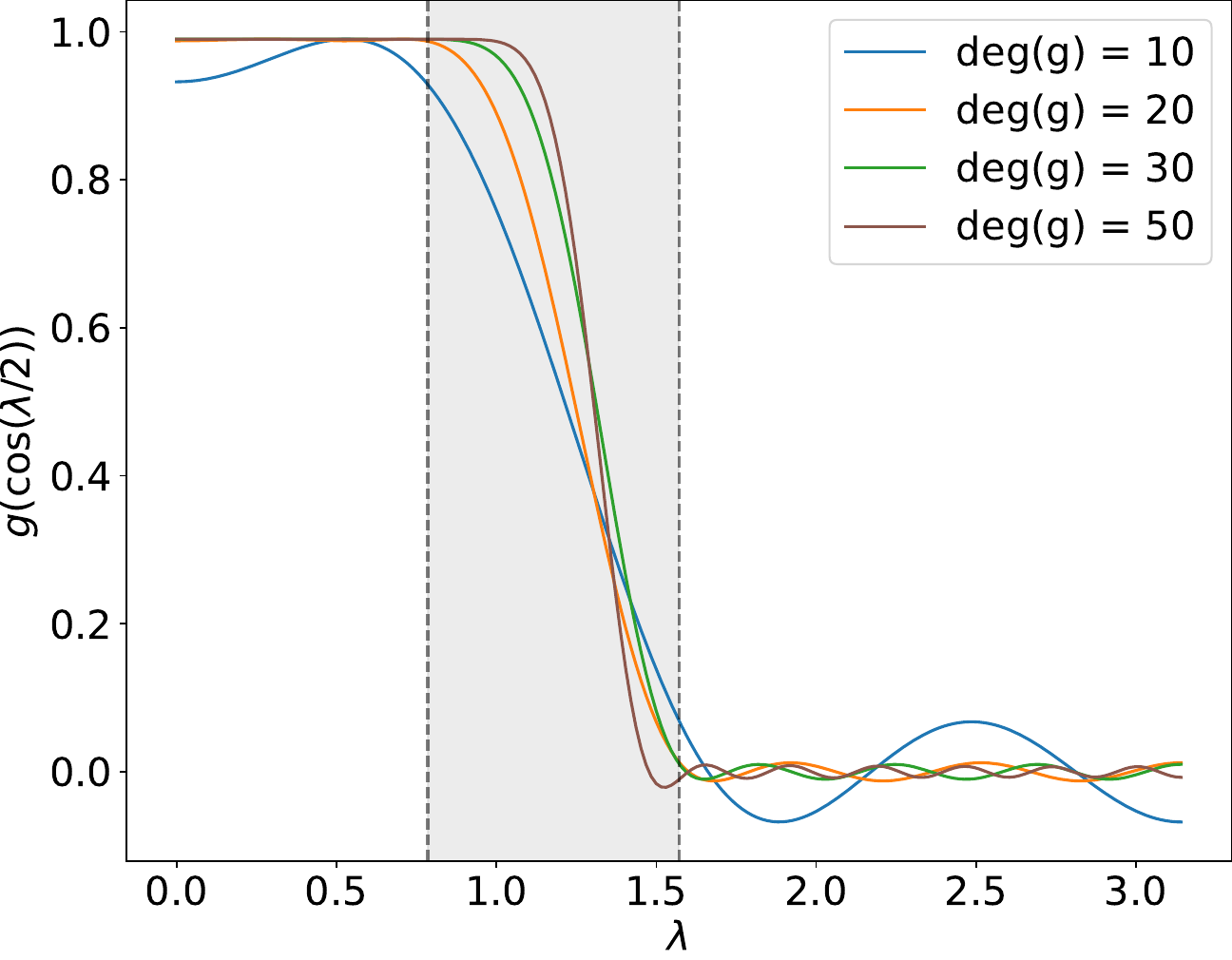}
    \caption{Numerical demonstration of the filter function. The dashed vertical lines stand for the boundary values $\lambda = \pi / 4$ and $\pi / 2$.}
    \label{fig:approx}
\end{figure}

Let us express the initial guess in the eigenbasis of $H$ as:
\begin{equation}
    \ket{\phi} = \sum_k \alpha_k \ket{\psi_k}.
\end{equation}

Then, applying $g(\cos(H t / 2))$ yields:
\begin{equation}
    g(\cos(H t / 2)) \ket{\phi} = \sum_k \alpha_k g(\cos(\lambda_k t / 2)) \ket{\psi_k} =: \sum_k \wt{\alpha}_k \ket{\psi_k} =: \wt{\phi}.
\end{equation}

Here, the ket notation is not applied to $\wt{\phi}$ in order to emphasize that it represents an unnormalized state. Given the properties of $g$, we observe:

\begin{enumerate}
    \item[(1)] For $0 \le \lambda_k < \pi / (4t)$, the coefficient remains almost unchanged: $\wt{\alpha}_k = g(\cos(\lambda_k t / 2)) \alpha_k \approx \alpha_k$.
    \item[(2)] For $\pi / (2t) < \lambda_k \le \pi / t$, the component is nearly eliminated: $\abs{\wt{\alpha}_k} \le \epsilon^\prime$.
    \item[(3)] In all cases, the magnitude of the component is monotonically non-increasing: $\abs{\wt{\alpha}_k} \le \abs{\alpha_k}$.
\end{enumerate}

By recurrently applying $g(\cos(H t_\ell / 2))$ with varying $t_\ell = 2^\ell / \norm{H}$, a sequence of unnormalized states
$$\left( \phi^{(0)}, \phi^{(1)}, \phi^{(2)}, \cdots, \phi^{(L)} \right)$$
is obtained with $\phi^{(0)} := \ket{\phi}$ and $\phi^{(\ell)} := g(\cos(H t_\ell / 2)) \phi^{(\ell - 1)}$. 

As the support of the quantum state $\phi^{(\ell)}$ at the $\ell$-th level is confined to less than $t_\ell \pi / 2$, the spectral radius of the effective Hamiltonian at this level is consequently upper bounded by $\pi \norm{H} / 2^{\ell + 1}$. The stopping criterion ensures that the effective spectral radius at the $L$-th stage is bounded by a constant, while preserving the amplitude of the ground state. Therefore, to achieve this bounded condition for the spectral radius of the effective Hamiltonian, we require:
\begin{equation*}
    \frac{\pi}{2^{L + 1}} \norm{H} \le \pi \Rightarrow L = \lceil \log_2(\norm{H} / 2) \rceil.
\end{equation*}
Furthermore, the condition that the ground-state amplitude is unaffected implies that it suffices to require
\begin{equation*}
    \lambda_0 \le \frac{\pi}{4} \le \frac{\pi}{4 t_L}
\end{equation*}
where $t_L = 2^L / \norm{H} \le 2^{\log_2(\norm{H} / 2) + 1} / \norm{H} = 1$ is used. Note that this condition can be easily achieved by shifting the spectrum of the Hamiltonian by applying $Z$-rotation gates when implementing the time evolution unitary.

Consequently, the cost for this sequential filtering is 
\begin{equation*}
    d \sum_{\ell = 1}^L \lceil t_\ell / \tau \rceil \le d \sum_{\ell = 1}^L (1 + t_\ell / \tau) = d L + d \tau^{-1} \Or(2^L \norm{H}^{-1}) = \Or(d ( L + \tau^{-1})).
\end{equation*}

It is worth noting that $t_\ell \le t_L = 2^L / \norm{H} = 1/ 2$ implies that the cost of sequential filtering is $\Or(d L)$ when $\tau \ge 1/2$ according to the derivation in the equation above.

Although the overall probability of success of the sequential filtering process is the product of the individual conditional probabilities of success at each level, the overall probability remains substantial as the conditional probability at each successive level increases. This intuition is quantified in the analysis that follows. 

Note that in the overlap between the last state and the ground state is 
\begin{equation}
    \abs{\braket{\psi_0 | \phi^{(L)}}} = \abs{\alpha_0 \prod_{j = 1}^L g(\cos(\lambda_0 t_j / 2))} \ge \abs{\alpha_0} \left( 1 - \abs{1 - \prod_{j = 1}^L g(\cos(\lambda_0 t_j / 2))} \right).
\end{equation}
Here, the first expression stands for the inner product between the ground state and the unnormalized state $\phi^{(L)}$. Note that $\abs{\alpha_0} = \gamma$ by definition. The last inequality follows triangle inequality which can be further bounded by telescoping the sequence:
\begin{equation}
    \begin{split}
        \abs{1 - \prod_{j = 1}^L g(\cos(\lambda_0 t_j / 2))} &= \abs{\sum_{k = 1}^L \left(1 - g(\cos(\lambda_0 t_k) / 2)\right) \prod_{j = 1}^{k - 1} g(\cos(\lambda_0 t_j / 2))}\\
        &\le \epsilon^\prime \sum_{k = 1}^L (1 + \epsilon^\prime)^{k - 1} = \left(1 + \epsilon^\prime\right)^L - 1.
    \end{split}
\end{equation}
We can choose $\epsilon^\prime L \le \ln(3/2)$ such that
\begin{equation}
    \left(1 + \epsilon^\prime\right)^L - 1 \le e^{L \epsilon^\prime} - 1 \le \frac{1}{2}.
\end{equation}
Then, we have a lower bound on the overlap 
\begin{equation}
    \abs{\braket{\psi_0 | \phi^{(L)}}} \ge \frac{1}{2} \gamma.
\end{equation}
To achieve the final target fidelity at least $1 - \epsilon$, the error parameter is $\epsilon^\prime = \gamma \epsilon$ where the factor $\gamma$ accounts for the subnormalization. Hence, the requirement is equivalent to $\gamma \epsilon L \le \ln(3/2)$. According to the choice of stopping point $L$, we deduce that it suffices to set $\epsilon \le \gamma^{-1} \log_2^{-1}(\norm{H} / 2) \ln(3 / 2)$ to fulfill the condition.

We then apply a different clean-up filter function to the state $\phi^{(L)}$. According to the previous discussion, the state $\phi^{(L)}$ approximately supports on the subspace with eigenvalue of the Hamiltonian lower than $\pi$. Unlike the sequential filters, this final filter is constructed by querying $e^{- \I H}$ for $\Or(\Delta^{-1} \log(\gamma^{-1} \epsilon^{-1}))$ times to achieve the target fidelity at least $1 - \epsilon$. Note that each $e^{- \I H}$ requires $\lceil \tau^{-1} \rceil$ queries to $O_H$ according to \cref{lma:fast-forward}. Here, the $\gamma$ dependency in logarithmic factor accounts for subnormalization that comes from post-selecting measurement results when analyzing the error. Consequently, the overlap between the ground state and the final unnormalized state is $\Omega(\gamma)$. A single run of this procedure costs
\begin{equation}
    \Or\left( d (L + \tau^{-1}) + \Delta^{-1} \lceil\tau^{-1} \rceil \log\left(\gamma^{-1} \epsilon^{-1}\right) \right) = \Or\left(\left(\log(\norm{H}) + (1 + \tau^{-1}) \Delta^{-1} \right) \log\left(\gamma^{-1} \epsilon^{-1}\right) \right).
\end{equation}
Here, the $\tau$-dependent term in the cost stands for the implementation cost due to the limited cutoff for fast-forwarding. When $\tau \ge 1$, the complexity becomes independent of $\tau$ according to the analysis, which is consistent with the outcome obtained by taking the limit as $\tau$ becomes sufficiently large in the above complexity scaling.

To prepare the ground state with constantly high success probability, one can either repeat experiments or implement an amplitude amplification to coherently boost the success probability. For the former choice, the action of each filter function can be implemented by intermediate measurement and post selection. This implementation particularly saves the use of ancilla qubits due to the reuse of ancilla qubit after measurements. However, an $\Or(1 / \gamma^2)$ number of repetitions is required because of the success probability due to the overlap. Consequently, the overall complexity is summarized in \cref{thm:multi-level}.

\begin{thm}[Multi-level QSP-based ground state preparation using fast-forwarded Hamiltonian evolution]\label{thm:multi-level}
    Under the requirement that $\epsilon \le \gamma^{-1} \log_2^{-1}(\norm{H} / 2) \ln(3 / 2)$ and $\norm{H} \ge \pi$, using multi-level QSP-based method, we can prepare the ground state, with probability $2/3$, up to fidelity $1 - \epsilon$, with the following cost:
    \begin{enumerate}
        \item[(1)] $\Or(\gamma^{-2})$ queries to $O_I$,
        \item[(2)] $\wt{\Or}\left(\gamma^{-2} \left(\log(\norm{H}) + (1 + \tau^{-1}) \Delta^{-1} \right) \log\left(\epsilon^{-1}\right) \right)$ gates,
        \item[(3)] one ancilla qubit.
    \end{enumerate}
\end{thm}

For the later choice of coherent implementation, we need to implement a single run of the procedure in a coherent quantum circuit in order to perform amplitude amplification. Note that a naive coherent implementation may directly use deferred measurements and it requires $\Or(L) = \Or(\log(\norm{H}))$ additional ancilla qubits. That makes the algorithm loose its advantage comparing with LCU-based implementation. An improved implementation can be constructed by using the compression gadget in Ref. \cite{LowWiebe2019,FangLinTong2023}. The key idea in this technique is using adder operators to keep track of the success or failure of each filter function. The adder operator implements the modular addition operation on a $m := \lceil \log_2(L + 2) \rceil$ qubit register as $\mathrm{ADD}\ket{i} = \ket{i + 1 \mod 2^m}$. This register serves as a counter. At initialization, the counter register is set to $\ket{L + 1}$ by applying $\mathrm{ADD}^{(L + 1)}$. Note that each level implements a filter function by using QETU, which is
\begin{equation}
    U_{\text{level } \ell} \ket{0}\ket{\phi^{(\ell - 1)}} = \underbrace{\ket{0} g(\cos(H t_\ell / 2)) \ket{\phi^{(\ell - 1)}}}_\text{success} + \underbrace{\ket{1} \ket{\text{garbage}}}_\text{failure}.
\end{equation}
Here, the success stands for the action of the filter function which is marked by the ancilla state $\ket{0}$. On the other hand, when the ancilla state is $\ket{1}$, the quantum state in the system register is unwanted and the case is considered as a failure.

The subsequent conjugated adder operator lowers the counter value once the output of a level is marked as a success by the ancilla state $\ket{0}$. Otherwise, the counter remains unchanged to account the failure. Because the maximal capacity of the counter is $2^m - 1 = 2^{\lceil \log_2(L+2) \rceil} - 1 \ge L + 1$, the counter ancilla state is $\ket{0^m}$ if and only if all levels are marked as success. 

Leveraging the compression gadget, when the spectral radius is large, there is a coherent quantum circuit using $\Or(\log L) = \Or(\log \log (\norm{H}))$ ancilla qubits and with gate complexity $\wt{\Or}((\log(\norm{H}) + (1 + \tau^{-1}) \Delta^{-1}) \log(\gamma^{-1} \epsilon^{-1})$, which generates a unnormalized state whose overlap with the ground state is $\Omega(\gamma)$ and the target fidelity is at least $1 - \epsilon$. Using amplitude amplification on this coherently implemented circuit, the overall complexity of ground state preparation is summarized in \cref{thm:multi-level-aa}. 
\begin{thm}[Multi-level QSP-based ground state preparation using fast-forwarded Hamiltonian evolution and amplitude amplification]\label{thm:multi-level-aa}
    Under the requirement that $\epsilon \le \gamma^{-1} \log_2^{-1}(\norm{H} / 2) \ln(3 / 2)$ and $\norm{H} \ge \pi$, using multi-level QSP-based method and amplitude amplification, we can prepare the ground state, with probability $2/3$, up to fidelity $1 - \epsilon$, with the following cost:
    \begin{enumerate}
        \item[(1)] $\Or(\gamma^{-1})$ queries to $O_I$,
        \item[(2)] $\wt{\Or}\left(\gamma^{-1} \left(\log(\norm{H}) + (1 + \tau^{-1}) \Delta^{-1}  \right) \log\left(\epsilon^{-1}\right) \right)$ gates,
        \item[(3)] $\Or(\log \log (\norm{H}))$ ancilla qubits.
    \end{enumerate}
\end{thm}

According to the result mentioned above, the gate complexity for any $\tau$ is $\wt{\Or}(\log(\norm{H}) + \Delta^{-1})$, illustrating an additive relation between the dependencies on the spectral gap and the logarithmic spectral radius. This formulation represents an improvement over the multiplicative term $\log(\norm{H})\Delta^{-1}$ observed in the gate complexity of the \selector oracle in the LCU-based method, as specified in \cref{thm:lcu-prepare}. Additionally, the number of required ancilla qubits in our multi-level QSP-based method shows an exponential improvement in terms of the $\norm{H}$dependency and notably lacks the undesirable dependencies on $\Delta$ and $\epsilon$ that are present in the LCU-based method's results, as detailed in \cref{thm:lcu-prepare}. 

One may think about designing a multi-level LCU-based algorithm for preparing the ground state in which each filter function is implemented by LCU. However, it is worth noting that the decay of the probability of success in each LCU-based level will eventually ruin the overall complexity due to the small accumulated overall probability of success.

\section{Ground state preparation methods using soft fast-forwarded Hamiltonian evolution}\label{sec:alpha-soft}
In the previous section, we discussed the algorithm and its complexity analysis with the application of $\tau$-cutoff fast-forwarded Hamiltonian evolution. When transitioning to $\alpha$-soft fast-forwarded Hamiltonian evolution, although the core concept and idea are the same, there is difference in the complexity analysis. This section will delve into a comprehensive analysis of the complexity that arise when changing the fast-forwarding model.

\subsection{LCU-based method}\label{sec:alpha-soft-lcu}
Because the \prepor oracle does not depend on the implementation of the Hamiltonian evolution oracle, the difference in the complexity analysis stems from that in the implementations of \selector oracle. When employing the $\alpha$-soft fast-forwarded Hamiltonian evolution, the cost associated with implementing the \selector oracle is as follows:
\begin{equation}
    2 \sum_{\ell = 0}^{\lceil \log_2(2d+1)\rceil} (2^\ell)^\alpha = \left\{
    \begin{array}{ll}
        \Or(\log(d)) & \text{, when } \alpha = 0, \\
        \Or(d^\alpha) & \text{, when } 0 < \alpha \le 1.
    \end{array}
    \right.
\end{equation}

When the amplitude amplification procedure is used, there is an additional multiplicative factor of $\Or(\gamma^{-1} \log(d))$ in the circuit depth, where the logarithmic term is due to the $1$-norm scaling $\norm{c}_1 = \Or(\log(d))$ in \cref{thm:Fourier-approx}. Therefore the total complexity of the ground state preparation for implementing the \selector oracle (with amplitude amplification) is:
\begin{enumerate}
    \item When $\alpha = 0$, the complexity is $\wt{\Or}(\gamma^{-1}\log^2(\norm{H} \Delta^{-1} \log(\epsilon^{-1})))$.
    \item When $0 < \alpha \le 1$, the complexity is $\Or(\gamma^{-1} d^\alpha \log(d)) = \wt{\Or}(\gamma^{-1} \norm{H}^\alpha \Delta^{-\alpha} \log^\alpha(\epsilon^{-1}))$.
\end{enumerate}

Integrating this with the analysis in \cref{sec:lcu}, we encapsulate the foregoing analysis into the subsequent theorem:
\begin{thm}[LCU-based ground state preparation using $\alpha$-soft fast-forwarded Hamiltonian evolution]\label{thm:alpha-soft-lcu}
    Using LCU-based method, we can prepare the ground state, with probability $2/3$, up to fidelity $1 - \epsilon$, with the following cost:
    \begin{enumerate}
        \item[(1)] $\Or(\gamma^{-1})$ queries to $O_I$,
        \item[(2a)] $\wt{\Or}(\gamma^{-1}\log^2(\norm{H} \Delta^{-1} \log(\epsilon^{-1})))$ gates to implement \selector oracle  when $\alpha = 0$,
        \item[(2b)] $\wt{\Or}(\gamma^{-1} \norm{H}^\alpha \Delta^{-\alpha} \log^\alpha(\epsilon^{-1}))$ gates to implement \selector oracle when $0 < \alpha \le 1$,
        \item[(3)] $\wt{\Or}(\gamma^{-1} \norm{H} \Delta^{-1} \log(\epsilon^{-1}))$ additional T gates to implement \prepor oracle,
        \item[(4)] $\wt{\Or}(\log(\norm{H} \Delta^{-1} \gamma^{-1} \epsilon^{-1}))$ ancilla qubits.
    \end{enumerate}
\end{thm}
It is worth noting that though the implementation cost of \selector oracle is improved by using $\alpha$-soft fast-forwarded Hamiltonian evolution and the required number of ancilla qubits depends weakly in the spectral radius in logarithm $\log(\norm{H})$, the implementation cost of \prepor oracle depends linearly in $\norm{H}$, which dominates the overall cost. We provide additional discussion of the complexity result (2b) as follows. For any $\alpha > 0$, the term $\norm{H}^\alpha \Delta^{- \alpha} \log^\alpha(\epsilon^{-1})$ is always faster than $\text{poly}\log(\norm{H}\Delta^{-1}\log(\epsilon^{-1}))$. As a consequence, it justifies the transition from the gate complexity in (2b) to that in (2a) as $\alpha \to 0^+$. Meanwhile, the ideal fast-forwarding can also be realized by setting $\tau$ to be larger than the intended simulation time. Under this circumstance, the gate complexity in \cref{thm:lcu-prepare} matches that in \cref{thm:alpha-soft-lcu} (2a).

When $\alpha \to 1$, the cost of Hamiltonian evolution becoming linear with time indicates the absence of fast-forwarding. This scenario is similar to setting the cutoff $\tau = \norm{H}^{-1}$ to be small in \cref{def:tau-cutoff-fast-forwarding}, which stands for small step size to construct long-time evolution. According to \cref{lma:fast-forward}, this small step size implies that the gate complexity synthesizing $e^{-\I t H}$ is $\Or(t \tau^{-1}) = \Or(t \norm{H})$ which is consistent with the scenario when $\alpha = 1$. Consequently, under this premise, the gate complexity results in the proceeding theorem align with that presented in \cref{thm:lcu-prepare}.

\subsection{Standard QSP-based method}\label{sec:alpha-soft-standard-qsp}
To employ the standard QSP-based method, as outlined in \cref{sec:standard-qsp}, the Hamiltonian is normalized to $H / \norm{H}$ for execution. This normalization results in the number of queries to $e^{-\I H / \norm{H}}$ be $\Or(\norm{H} \Delta^{-1} \log(\gamma^{-1} \epsilon^{-1}))$ due to the shrinkage in the relative spectral gap. By considering the implementation cost of the time evolution as specified in \cref{def:tau-cutoff-fast-forwarding}, we compile the cost associated with preparing the ground state using the standard QSP-based method and amplitude amplification as follows.

\begin{thm}[QSP-based ground state preparation]\label{thm:alpha-soft-standard-qsp}
    Using a QSP-based method, we can prepare the ground state, with probability $2/3$, up to fidelity $1 - \epsilon$, with the following cost:
    \begin{enumerate}
        \item[(1)] $\Or(\gamma^{-1})$ queries to $O_I$,
        \item[(2)] $\wt{\Or}(\gamma^{-1} \norm{H} \Delta^{-1} \log(\epsilon^{-1}))$ gates,
        \item[(3)] two ancilla qubits.
    \end{enumerate}
\end{thm}
The gate complexity of the standard QSP-based method sees no advantage from utilizing $\alpha$-soft fast-forwarded Hamiltonian evolution. This stems from the fact that the building block of this method involves only short-time evolution, the cost of which remains largely unaffected by fast-forwarding capabilities.

For any other $0 \le \alpha < 1$, the gate complexity in the implementation cost of the \selector oracle in LCU-based method (see \cref{thm:alpha-soft-lcu}) is found to be less than the gate complexity for the standard QSP-based method as outlined in \cref{thm:alpha-soft-standard-qsp}. Given the concave nature of fast-forwarded Hamiltonian evolution's gate complexity with respect to evolution time, exemplified by $t_1^\alpha + t_2^\alpha \ge (t_1 + t_2)^\alpha$, there is an incentive to employ longer-time Hamiltonian evolution in a layered approach to minimize overall gate complexity. This insight motivates the development of the multi-level QSP-based method, which capitalizes on the \selector oracle's functionality.

\subsection{Multi-level QSP-based method}\label{sec:alpha_soft_multi_level_qsp}
In line with the approach detailed in \cref{sec:multi-qsp}, the multi-level QSP-based method aims to reduce the effective spectral radius of the Hamiltonian through a sequential, multi-level application of QSP-based filters. In contrast to the case of small cutoff time $\tau$, the case with sufficiently large cutoff time facilitates a more significant reduction in the effective spectral radius, achieved through higher resolution in the sequential filtering process. To be precise, the sequential filtering process can continue until the exact ground state is prepared, eliminating the need for an additional filter function at the end. As we will demonstrate in the analysis in this subsection, the dependency on the spectral gap manifests logarithmically rather than as an additive term outside the logarithm in the case with small cutoff time. 

Similar to the discussion in \cref{sec:multi-qsp}, a sequence of unnormalized states can be obtained by recurrently applying $g(\cos(H t_\ell / 2))$ with varying evolution time $t_\ell = 2^\ell / \norm{H}$: 
$$\left( \phi^{(0)}, \phi^{(1)}, \phi^{(2)}, \cdots, \phi^{(\wt{L})} \right)$$
where $\phi^{(0)} := \ket{\phi}$ and $\phi^{(\ell)} := g(\cos(H t_\ell / 2)) \phi^{(\ell - 1)}$. The stopping point $\wt{L} = \lceil \log_2(\pi \norm{H} / (4 \Delta)) \rceil$ is chosen to eliminate all excited-state components at the final level. Note that the evolution time at the final level is bounded as $\pi / (4 \Delta) \le t_{\wt{L}} \le \pi / (2 \Delta)$.
The spectral radius of the effective Hamiltonian at the $\wt{L}$-th level is reduced to
\begin{equation*}
    \norm{H_\text{eff}} \le \frac{\pi}{2 t_{\wt{L}}} \le 2 \Delta
\end{equation*}
which is smaller than the $\Or(1)$ effective spectral radius thanks to the large cutoff time.

Meanwhile, the condition that the ground-state component is almost unchanged requires 
\begin{equation*}
    \lambda_0 \le \frac{\Delta}{2} \le \frac{\pi}{4 t_{\wt{L}}}.
\end{equation*}
Note that this condition can be easily achieved by shifting the spectrum of the Hamiltonian by applying $Z$-rotation gates.

The cost of the sequential filtering is 
\begin{equation*}
    \sum_{\ell = 1}^{\wt{L}} \left(2^\ell \right)^\alpha \Or(\log(\gamma^{-1} \epsilon^{-1})) = \left\{
    \begin{array}{ll}
        \Or(\log(\norm{H} \Delta^{-1}) \log(\gamma^{-1} \epsilon^{-1})) & \text{, when } \alpha = 0, \\
        \Or(\norm{H}^\alpha \Delta^{- \alpha} \log(\gamma^{-1} \epsilon^{-1})) & \text{, when } 0 < \alpha \le 1.
    \end{array}
    \right.
\end{equation*}
This improves the linear dependency of $\norm{H}$ in the standard QSP-based method thanks to the reduction of the effective spectral radius. 

The requirement $\abs{\braket{\psi_0 | \phi^{(\wt{L})}}} \ge \gamma / 2$ imposes an additional condition $\gamma \epsilon \wt{L} \le \ln(3 / 2)$ according to the analysis in \cref{sec:multi-qsp}, which is equivalent to $\epsilon \le \gamma^{-1} \log_2^{-1}(\pi \norm{H} / (4 \Delta)) \ln(3/2)$. 

When intermediate measurements are used, there is a multiplicative factor $\Or(\gamma^{-2})$ contributing to the overall gate complexity due to the repetition of quantum circuits. The cost is summarized below.
\begin{thm}[Multi-level QSP-based ground state preparation using $\alpha$-soft fast-forwarded Hamiltonian evolution]\label{thm:alpha-soft-multi-level}
    Under the requirement that $\epsilon \le \gamma^{-1} \log_2^{-1}(\pi \norm{H} / (4 \Delta)) \ln(3/2)$, using multi-level QSP-based method, we can prepare the ground state, with probability $2/3$, up to fidelity $1 - \epsilon$, with the following cost:
    \begin{enumerate}
        \item[(1)] $\Or(\gamma^{-2})$ queries to $O_I$,
        \item[(2a)] $\wt{\Or}(\gamma^{-2}\log(\norm{H} \Delta^{-1}) \log(\epsilon^{-1}))$ gates when $\alpha = 0$,
        \item[(2a)] $\wt{\Or}(\gamma^{-2}\norm{H}^\alpha \Delta^{-\alpha} \log(\epsilon^{-1}))$ gates when $0 < \alpha \le 1$,
        \item[(3)] one ancilla qubit.
    \end{enumerate}
\end{thm}

The dependence in the initial overlap $\gamma$ can be further improved by using amplitude amplification and compression gadget (see \cref{sec:multi-qsp} for details). The tradeoff is the number of ancilla qubits is not a constant in order to perform the multi-level quantum circuit coherently. The cost is summarized below.

\begin{thm}[Multi-level QSP-based ground state preparation using $\alpha$-soft fast-forwarded Hamiltonian evolution and amplitude amplification]\label{thm:alpha-soft-multi-level-aa}
    Under the requirement that $\epsilon \le \gamma^{-1} \log_2^{-1}(\pi \norm{H} / (4 \Delta)) \ln(3/2)$, using multi-level QSP-based method, we can prepare the ground state, with probability $2/3$, up to fidelity $1 - \epsilon$, with the following cost:
    \begin{enumerate}
        \item[(1)] $\Or(\gamma^{-1})$ queries to $O_I$,
        \item[(2a)] $\wt{\Or}(\gamma^{-1}\log(\norm{H} \Delta^{-1}) \log(\epsilon^{-1}))$ gates when $\alpha = 0$,
        \item[(2a)] $\wt{\Or}(\gamma^{-1} \norm{H}^\alpha \Delta^{-\alpha} \log(\epsilon^{-1}))$ gates when $0 < \alpha \le 1$,
        \item[(3)] $\Or(\log\log(\norm{H} \Delta^{-1}))$ ancilla qubits.
    \end{enumerate}
\end{thm}

Though the implementation cost of the \selector oracle in LCU-based method sees improvements in both $\norm{H}$ and $\Delta$ due to the $\alpha$-soft fast-forwarding assumption, the predominant contribution to the overall gate complexity comes from the costly implementation of the \prepor oracle, which incurs a complexity of $\wt{\Or}(\norm{H} \Delta^{-1})$. However, the efficient implementation of \selector oracle with fast-forwarded Hamiltonian evolution underpins the development of a multi-level QSP-based method. This approach benefits from a multi-level filtering strategy enabled by fast-forwarded Hamiltonian evolution. In ideal fast-forwarding, it leads to a method where the overall gate complexity demonstrates an exponential improvement with respect to both the spectral radius and the spectral gap. Moreover, the requirement for ancilla qubits is significantly reduced in comparison to the LCU-based method, courtesy of employing a compression gadget.

\section{Conclusion}\label{sec:conclusion}

This work studies the ground state preparation problem with fast-forwarded Hamiltonian evolution access. Coupled with a binary search algorithm~\cite{LinTong2020a,DongLinTong2022}, the algorithms developed in this work can also be used for ground-state energy estimation.
In the idealized fast-forwarding regime, the cost of the implementation of the \selector oracle in LCU based methods exhibits logarithmic dependence on the spectral radius. However, implementing the \prepor oracle introduces a linear scaling in spectral radius in the gate count. The standard QSP-based matrix function transformation techniques cannot efficiently make use of the fast-forwarding assumptions. We overcome this problem by developing a multi-level QSP method, which integrates fast-forwarded Hamiltonian evolution with the QETU algorithm~\cite{DongLinTong2022}. In the case of ideal fast-forwarding, our multi-level QSP-based approach demonstrates exponential improvements in both gate complexity and the number of ancilla qubits relative to the LCU-based method in terms of spectral radius dependency. 
Consequently, it bridges the gap in complexity dependence on the spectral radius between the \selector oracle in LCU-based method and QSP-based method. While our presentation focuses on the ground state, a minor modification of our multi-level QSP-based method enables the preparation of quantum states supported in a low-energy subspace. It would also be very interesting to explore how the multi-level QSP-based method could be adapted to efficiently compute matrix functions beyond the shifted sign function used for ground state preparation.

\section*{Acknowledgments}

This material is based upon work supported by the U.S. Department of Energy, Office of Science, National Quantum Information Science Research Centers, Quantum Systems Accelerator  (YD).  Additional funding is provided by the Challenge Institute for Quantum Computation (CIQC) funded by NSF through grant number OMA-2016245 (LL). LL is a Simons Investigator in Mathematics. The authors thank Dong An, Zhiyan Ding, Jin-Peng Liu, Murphy Yuezhen Niu, Yu Tong for helpful suggestions.

\bibliographystyle{abbrvnat}

\appendix

\section{Polynomial matrix functions using QETU}\label{app:qetu-derivation}

We begin by examining the circuit structure depicted in \cref{fig:qetu} to justify the polynomial function transformation. The analysis conducted in this section revisits the findings presented in Ref. \cite{DongLinTong2022}. 

\begin{figure}[htbp]
    \centering
    \includegraphics[width = .35\textwidth]{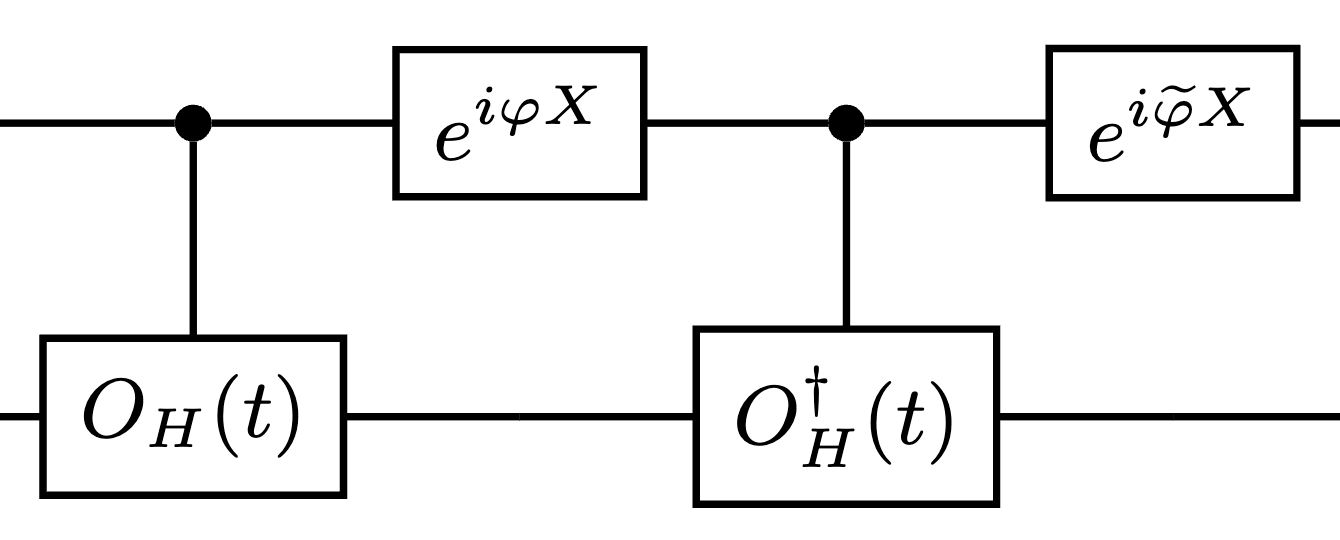}
    \caption{Building block of the \QETU circuit (see \cref{fig:qetu} for details).}
    \label{fig:QETU_building_block}
\end{figure}

The building block of the QETU circuit in \cref{fig:qetu} is depicted in \cref{fig:QETU_building_block} and incorporates two adjustable phase factors $(\varphi, \wt{\varphi})$. When the input quantum state $\ket{\psi}$ of the system registers is an eigenstate of the Hamiltonian with eigenvalue $\lambda$, namely, $H \ket{\psi} = \lambda \ket{\psi}$, this building block can be represented as a two-dimensional matrix in the basis $\{ \ket{0}\ket{\psi}, \ket{1}\ket{\psi}\}$:
\begin{equation}
\begin{split}
    e^{\I \wt{\varphi} X} \begin{pmatrix}
        1 & 0 \\ 0 & e^{\I t \lambda}
    \end{pmatrix} e^{\I \varphi X} \begin{pmatrix}
        1 & 0 \\ 0 & e^{- \I t \lambda}
    \end{pmatrix} &= e^{\I \wt{\varphi} X} e^{- \I t \lambda / 2 Z} e^{\I \varphi X} e^{\I t \lambda / 2 Z}\\
    &= e^{\I (\wt{\varphi} - \pi/2) X} \I X e^{- \I t \lambda / 2 Z} \I X e^{\I (\varphi - \pi / 2) X} e^{\I t \lambda / 2 Z}\\
    &= - e^{\I (\wt{\varphi} - \pi/2) X} e^{- \I t \lambda / 2 X Z X} e^{\I (\varphi - \pi / 2) X} e^{\I t \lambda / 2 Z}\\
    &=  - e^{\I (\wt{\varphi} - \pi/2) X} e^{\I t \lambda / 2 Z} e^{\I (\varphi - \pi / 2) X} e^{\I t \lambda / 2 Z}
\end{split}
\end{equation}
where the anti-commuting relation $XZX = - Z$ is used.

Suppose there are $d + 1$ phase factors in the QETU circuit in \cref{fig:qetu} with $d$ an even number. Then, the above equation implies that the QETU circuit admits the following matrix representation under the basis $\{ \ket{0}\ket{\psi}, \ket{1}\ket{\psi}\}$:
\begin{equation}
\begin{split}
    & (-1)^{d / 2} e^{\I (\varphi_0 - \pi/2) X} e^{\I t \lambda / 2 Z} e^{\I (\varphi_1 - \pi/2) X} e^{\I t \lambda / 2 Z}  \cdots e^{\I t \lambda / 2 Z} e^{\I (\varphi_1 - \pi/2) X} e^{\I t \lambda / 2 Z} e^{\I \varphi_0 X} = \\
    & (-1)^{d / 2} e^{- \I \pi / 4 X} e^{\I (\varphi_0 - \pi/4) X} e^{\I t \lambda / 2 Z} e^{\I (\varphi_1 - \pi/2) X} e^{\I t \lambda / 2 Z}  \cdots e^{\I t \lambda / 2 Z} e^{\I (\varphi_1 - \pi/2) X} e^{\I t \lambda / 2 Z} e^{\I (\varphi_0 - \pi / 4) X} e^{\I \pi / 4 X}.
\end{split}
\end{equation}
Let us consider another set of phase factors defined by
\begin{equation}\label{eqn:phase-factor-connection}
    \phi_j := \varphi_j - (2 - \delta_{j0})\pi / 4, \text{ for } j = 0, \cdots, d / 2.
\end{equation}
By conjugating the representation by Hadamard gates acting on the ancilla qubit, the roles of Pauli X- and Z-matrices are switched:
\begin{equation}
    (-1)^{d / 2} \mathrm{H} e^{- \I \pi / 4 Z} [e^{\I \phi_0 Z} e^{\I t \lambda / 2 X} e^{\I \phi_1 Z} e^{\I t \lambda / 2 X}  \cdots e^{\I t \lambda / 2 X} e^{\I \phi_1 Z} e^{\I t \lambda / 2 X} e^{\I \phi_0 Z}] e^{\I \pi / 4 Z} \mathrm{H}.
\end{equation}
It is worth noting that the set of phase factors is symmetric because $\phi_j = \phi_{d - j}$ for any $j = 0, \cdots, d / 2$. According to \cite[Theorem 1]{WangDongLin2022}, there exists real polynomials $g, h, q \in \RR[x]$ with $\deg(g) = d$ so that the matrix product in the square bracket satisfies
\begin{equation}
    [\cdot] = \begin{pmatrix}
        g(x) + \I h(x) & \I q(x) \sqrt{1 - x^2} \\ 
        \I q(x) \sqrt{1 - x^2} & g(x) - \I h(x)
    \end{pmatrix}
\end{equation}
where $x := \cos(t \lambda / 2)$. Using this result, the two-dimensional matrix representation of the QETU circuit is
\begin{equation}
    (-1)^{d / 2} \mathrm{H} e^{- \I \pi / 4 Z} [\cdot]  e^{\I \pi / 4 Z} \mathrm{H} = (-1)^{d / 2} \begin{pmatrix}
        g(x) & - q(x) \sqrt{1-x^2} + \I h(x) \\
        q(x) \sqrt{1-x^2} + \I h(x)  & g(x)
    \end{pmatrix}.
\end{equation}
Suppose the input quantum state is not an eigenstate and $\ket{\psi} = \sum_j \alpha_j \ket{\psi_j}$, due to linearity and up to a global phase $(-1)^{d / 2}$, the quantum state after acting the QETU circuit is
\begin{equation}
    \ket{0} \sum_j \alpha_j g(\cos(t \lambda_j / 2)) \ket{\psi_j} + \ket{1}\ket{*} = \ket{0} g(\cos(t H / 2)) \ket{\psi} + \ket{1}\ket{*}.
\end{equation}
Here, $\ket{*}$ stands for some unnormalized quantum state that is not of our interest. Hence, by measuring the ancilla qubit with $0$, the QETU circuit transforms the input state according to the matrix function $g(\cos(t H / 2))$.

When a target polynomial $g$ is given, the set of phase factors used in the QETU circuit can be derived as follows. First, we derive a set of symmetric phase factors $\Phi := (\phi_0, \phi_1, \cdots, \phi_1, \phi_0) \in \RR^{d + 1}$ corresponding to the target polynomial $g$ using $\mathrm{QSPPACK}$ software package, which facilitates the computation of these factors \footnote{The examples are available on the website \url{https://qsppack.gitbook.io/qsppack/} and the codes are open-sourced in \url{https://github.com/qsppack/QSPPACK}.}\cite{DongMengWhaleyEtAl2021}. Subsequently, the required phase factors for use in the QETU circuit are calculated based on the relationship specified in \cref{eqn:phase-factor-connection}. In \cref{tab:phase-factor}, we present a set of degree-20 phase factors that were derived using the aforementioned procedure. These phase factors are utilized for approximating the filter function, as shown in \cref{fig:approx}. Additionally, \cref{fig:approx_err} provides a visualization of both the filter function and its pointwise approximation error.

\begin{table}[htbp]
\centering
\begin{tabular}{@{} *{7}{c} @{}}\midrule
\multicolumn{7}{c@{}}{$L^\infty$ error $\displaystyle := \max\left\{\max_{\cos(\pi/8) \le x \le 1} \abs{g(x) - 1}, \max_{0 \le x \le \cos(\pi/4)} \abs{g(x)}\right\} = 0.01333$}\\\midrule
$\varphi_{0}$ & $\varphi_{1}$ & $\varphi_{2}$ & $\varphi_{3}$ & $\varphi_{4}$ & $\varphi_{5}$ & $\varphi_{6}$ \\
1.5641113 & 1.5804045 & 1.5942229 & 1.5741280 & 1.5233379 & 1.5189284 & 1.6198455 \\\midrule
$\varphi_{7}$ & $\varphi_{8}$ & $\varphi_{9}$ & $\varphi_{10}$ & $\varphi_{11}$ & $\varphi_{12}$ & $\varphi_{13}$ \\
1.7237235 & 1.5881872 & 1.1064466 & 0.7862644 & 1.1064466 & 1.5881872 & 1.7237235 \\\midrule
$\varphi_{14}$ & $\varphi_{15}$ & $\varphi_{16}$ & $\varphi_{17}$ & $\varphi_{18}$ & $\varphi_{19}$ & $\varphi_{20}$ \\
1.6198455 & 1.5189284 & 1.5233379 & 1.5741280 & 1.5942229 & 1.5804045 & 1.5641113 \\\midrule
\end{tabular}
\caption{A set of phase factors used in \cref{fig:approx}.}
\label{tab:phase-factor}
\end{table}

\begin{figure}[htbp]
    \centering
    \includegraphics[width=\textwidth]{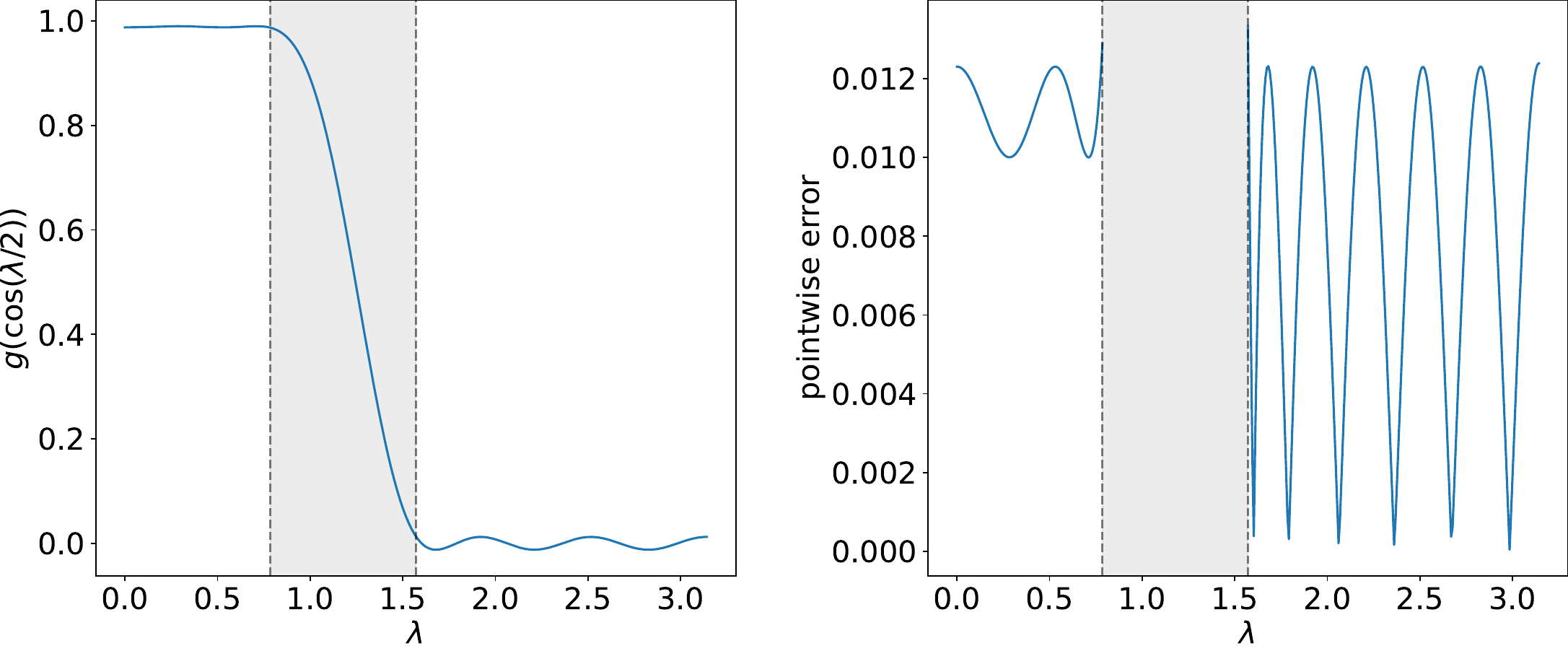}
    \caption{Visualization of the filter function approximation (left) and its pointwise approximation error (right). The error is quantified as $\abs{g(\cos(\lambda / 2)) - 1}$ for $0 \le \lambda \le \pi / 4$, and $\abs{g(\cos(\lambda / 2))}$ for $\pi / 2 \le \lambda \le \pi$. This approximation uses the phase factors listed in \cref{tab:phase-factor}. The dashed vertical lines stand for the boundary values $\lambda = \pi / 4$ and $\pi / 2$.}
    \label{fig:approx_err}
\end{figure}

\section{Propagation of the implementation error of Hamiltonian evolution}\label{sec:HE-error}
In practice, the Hamiltonian evolution is implemented with bounded error as in \cref{def:tau-cutoff-fast-forwarding}. Consequently, analyzing the propagation of the implementation error of Hamiltonian evolution is helpful for understanding the overall performance of the algorithm. For the purposes of this discussion, we will simplify the analysis by considering a circuit model that utilizes $q$ queries to the Hamiltonian evolution operator. Specifically
\begin{align}
    & U_\text{HE} := G_0 e^{- \I t_1 H} G_1 e^{- \I t_2 H} G_2 \cdots G_{q - 1} e^{- \I t_q H} G_q.\\
    & U_\text{OH} := G_0 O_H(t_1) G_1 O_H(t_2) G_2 \cdots G_{q - 1} O_H(t_q) G_q,
\end{align}
where $t_1, \cdots, t_q \le \tau$. Here, $G_0, \cdots, G_q$ stand for some gate operations that are independent with the Hamiltonian evolution. By telescoping, the difference between these two unitaries are bounded:
\begin{equation}
    \norm{U_\text{HE} - U_\text{OH}} \le \sum_{j = 1}^q \norm{e^{- \I t_j H} - O_H(t_j)} \le q \delta.
\end{equation}

To illustrate this, we perform a rough analysis using the derived linear growth error bound when implementing Hamiltonian evolution via the Trotter formula. Considering that the evolution time parameter used in our multi-level QSP algorithm is less than one. Hence, it is reasonable to simplify the analysis under the assumption that
\begin{equation*}
    \sup_{|t| \le 1} \norm{e^{-\I t H} - O_H(t)} \le \delta \le C_T r^{- p}
\end{equation*}
where a $p$-th order Trotter formula is used with $r$ Trotter steps. Here, $C_T$ denotes a prefactor which involves the commutators among Hamiltonian terms in $H$ standing for fast-forwarding. In the case of ideal fast-forwarding, to achieve a target precision $\epsilon$, the total number of queries to the Hamiltonian evolution is $q = \wt{\Or}(\gamma^{-1} \log(\norm{H} \Delta^{-1}) \log(\epsilon^{-1}))$ according to \cref{thm:alpha-soft-multi-level-aa}. Inverting the error bound, in order to achieve a target precision $\epsilon$, the Trotter number is chosen to
\begin{equation}
    r = \wt{\Or}\left( C_T^{1/p} \log^{1/p}(\norm{H}\Delta^{-1}) \gamma^{-1/p} \epsilon^{-1/p} \right).
\end{equation}

\end{document}

%% file: preamble.tex
\usepackage{color, xcolor, colortbl}
\usepackage{graphicx}
\usepackage{geometry}
\usepackage{amsthm,amsmath,amssymb,amsfonts}
\usepackage{dcolumn}
\usepackage{url}
\usepackage{epstopdf}
\usepackage{enumitem}
\usepackage[section]{algorithm}
\usepackage{algpseudocode}
\usepackage{bm}
\usepackage[caption=false]{subfig}
\usepackage{appendix}
\usepackage{multirow}
\usepackage{braket}
\usepackage[english]{babel}
\usepackage{hyperref}
\usepackage[capitalize]{cleveref}
\usepackage{xpatch}
\usepackage{tikz}
\usepackage{adjustbox}
\usepackage{xspace}

\usepackage{dsfont}

\usepackage{multirow}
\newcommand{\I}{\mathrm{i}}

\newcommand{\mc}[1]{\mathcal{#1}}

\newcommand{\wt}[1]{\widetilde{#1}}

\newcommand{\abs}[1]{\left\lvert#1\right\rvert}
\newcommand{\norm}[1]{\left\lVert#1\right\rVert}

\newcommand{\Or}{\mathcal{O}}

\newcommand{\RR}{\mathbb{R}}

\newcommand{\ZZ}{\mathbb{Z}}

\newcommand{\QETU}{{QETU}\xspace}

\newtheorem{thm}{\protect\theoremname}
\newtheorem{lem}[thm]{\protect\lemmaname}

\newtheorem{defn}[thm]{\protect\definitionname}

\providecommand{\definitionname}{Definition}
\providecommand{\assumptionname}{Assumption}
\providecommand{\corollaryname}{Corollary}
\providecommand{\lemmaname}{Lemma}
\providecommand{\propositionname}{Proposition}
\providecommand{\remarkname}{Remark}
\providecommand{\theoremname}{Theorem}
\usepackage{bbm}



\usetikzlibrary{fit}
\tikzset{%
  highlight/.style={rectangle,rounded corners,fill=blue!15,draw,fill opacity=0.3,thick,inner sep=0pt}
}

